\documentclass[12pt]{article}
\usepackage{amssymb,amsmath,amsthm,enumerate,comment}

\usepackage{graphicx, color}

\usepackage{geometry}
\usepackage{epstopdf}
\usepackage{setspace}
\AtBeginDocument{%
  \addtolength\abovedisplayskip{-0.25\baselineskip}%
  \addtolength\belowdisplayskip{-0.15\baselineskip}%
}

\usepackage{float}
\floatstyle{ruled}

\usepackage[titletoc,title]{appendix}
\usepackage{url}

\usepackage{paralist, enumitem}

\usepackage{nicefrac}

\usepackage{amsfonts, verbatim}
\usepackage{hyphenat,epsfig,subfigure,multirow}

\usepackage[usenames,dvipsnames]{xcolor}
\usepackage[ruled]{algorithm2e}

\usepackage{tcolorbox}
\tcbuselibrary{skins,breakable}
\tcbset{enhanced jigsaw}

\usepackage[compact]{titlesec}

\definecolor{DarkRed}{rgb}{0.5,0.1,0.1}
\definecolor{DarkBlue}{rgb}{0.1,0.1,0.5}

\usepackage{nameref}
\definecolor{ForestGreen}{rgb}{0.1333,0.5451,0.1333}
\definecolor{Red}{rgb}{0.9,0,0}
\usepackage[linktocpage=true,
	pagebackref=true,colorlinks,
	linkcolor=DarkRed,citecolor=ForestGreen,
	bookmarks,bookmarksopen,bookmarksnumbered]
	{hyperref}
\usepackage[noabbrev,nameinlink]{cleveref}
\crefname{property}{property}{Property}
\creflabelformat{property}{(#1)#2#3}
\crefname{equation}{eq}{Eq}
\creflabelformat{equation}{(#1)#2#3}

\usepackage{tikz}
\usetikzlibrary{arrows}
\usetikzlibrary{arrows.meta}
\usetikzlibrary{shapes}
\usetikzlibrary{backgrounds}
\usetikzlibrary{positioning}
\usetikzlibrary{decorations.markings}
\usetikzlibrary{patterns}
\usetikzlibrary{calc}
\usetikzlibrary{fit}
\usetikzlibrary{snakes}

\usepackage{bm, bbm, mathrsfs}
\usepackage{url}
\usepackage{xspace}

\usepackage{mdframed}

\newtheorem*{mdresult}{Main Result}
\newenvironment{result}{\begin{mdframed}[backgroundcolor=lightgray!40,topline=false,bottomline=false, innerbottommargin=12pt]\begin{mdresult}}{\end{mdresult}\end{mdframed}}

\usepackage{aliascnt} 

\newtheorem{theorem}{Theorem}[section]
\newtheorem*{theorem*}{Theorem}

\newaliascnt{definition}{theorem}
\newtheorem{definition}[definition]{Definition}
\aliascntresetthe{definition}

\newtheorem*{definition*}{Definition}

\newaliascnt{lemma}{theorem}
\newtheorem{lemma}[lemma]{Lemma}
\aliascntresetthe{lemma}

\newtheorem*{lemma*}{Lemma}

\newaliascnt{claim}{theorem}

\aliascntresetthe{claim}

\newtheorem*{claim*}{Claim}

\newaliascnt{fact}{theorem}
\newtheorem{fact}[fact]{Fact}
\aliascntresetthe{fact}

\newtheorem*{fact*}{Fact}

\newaliascnt{observation}{theorem}
\newtheorem{observation}[observation]{Observation}
\aliascntresetthe{observation}

\newtheorem*{observation*}{Observation}

\newaliascnt{conjecture}{theorem}

\aliascntresetthe{conjecture}

\newtheorem*{conjecture*}{Conjecture}

\newaliascnt{corollary}{theorem}
\newtheorem{corollary}[corollary]{Corollary}
\aliascntresetthe{corollary}

\newtheorem*{corollary*}{Corollary}

\newaliascnt{remark}{theorem}

\aliascntresetthe{remark}

\newtheorem*{remark*}{Remark}

\newaliascnt{proposition}{theorem}

\aliascntresetthe{proposition}

\newtheorem*{proposition*}{Proposition}

\allowdisplaybreaks

\newcommand{\tvd}[2]{\ensuremath{\norm{#1 - #2}_{tvd}}}

\newcommand{\eps}{\ensuremath{\varepsilon}}

\newcommand{\norm}[1]{\ensuremath{\|#1\|}}

\newcommand{\poly}{\mbox{\rm poly}}

\newcommand{\opt}{\textnormal{\ensuremath{\mbox{opt}}}\xspace}

\DeclareMathOperator*{\Prob}{\ensuremath{\textnormal{Pr}}}
\renewcommand{\Pr}{\Prob}

\newenvironment{tbox}{\begin{tcolorbox}[
		enlarge top by=5pt,
		enlarge bottom by=5pt,
		 breakable,
		 boxsep=0pt,
                  left=4pt,
                  right=4pt,
                  top=10pt,
                  arc=0pt,
                  boxrule=1pt,toprule=1pt,
                  colback=white
                  ]
	}
{\end{tcolorbox}}

\newcommand{\kl}[2]{\ensuremath{\mathbb{D}(#1~||~#2)}}

\newcommand{\distribution}[1]{\ensuremath{\textnormal{dist}(#1)}\xspace}

\newcommand{\unif}{\ensuremath{\mathcal{U}}}

\renewcommand{\bar}[1]{\overline{#1}}

\newcommand{\REM}[1]{}

\def\part{\mathsf{Part}}
\def\unif{\mathcal{U}}

\def\alice{\mathsf{A}}
\def\bob{\mathsf{B}}

\def\seller{\mathsf{S}}

\def\alloc{\mathsf{alloc}}
\def\price{\mathsf{price}}
\def\cc{\mathsf{CC}}
\def\opt{\mathsf{opt}}

\def\bxos{\mathtt{BXOS}}
\def\xos{\mathtt{XOS}}

\def\supp{\mathsf{supp}}
\def\len{\mathsf{len}}

\def\I{\mathbb{I}}
\def\H{\mathbb{H}}

\def\vecopt{\vec{\mathsf{opt}}}
\def\vecpair{\vec{\mathsf{pair}}}

\def\vecreg{\vec{\mathsf{reg}}}
\def\vecregpair{\vec{\mathsf{regpair}}}
\def\vecspec{\vec{\mathsf{spec}}}
\def\vecspecpair{\vec{\mathsf{specpair}}}
\def\veccmp{\vec{\mathsf{cmp}}}

\def\E{\mathop{\mathbb{E}}}

\def\ally{\mathsf{PC\text{-}ally}}
\def\pc{\mathsf{PC}}

\title{Separating the Communication Complexity of Truthful and Non-Truthful Combinatorial Auctions}
\author{Sepehr Assadi\footnote{Department of Computer Science, Rutgers University. Email: \texttt{sepehr.assadi@rutgers.edu}. Part of this work was done while the author was a postdoctoral researcher at Princeton University and was supported in part
by the Simons Collaboration on Algorithms and Geometry.} 
\and Hrishikesh Khandeparkar\footnote{Department of Computer Science, Princeton University. Email: \texttt{hrishikesh.khandeparkar@gmail.com}.} 
\and Raghuvansh R. Saxena\footnote{Department of Computer Science, Princeton University. Email:  \texttt{rrsaxena@cs.princeton.edu}. Research supported by the National Science Foundation CAREER award CCF-1750443. }
\and S. Matthew Weinberg\footnote{Department of Computer Science, Princeton University. Email: \texttt{smweinberg@princeton.edu}. Supported by NSF CCF-1717899.} 
}

\date{}

\begin{document}
\maketitle

\pagenumbering{roman}

\begin{abstract}
	We provide the first separation in the approximation guarantee achievable by truthful and non-truthful combinatorial auctions with polynomial communication. Specifically, we prove that any truthful mechanism guaranteeing a $(\nicefrac{3}{4}-\nicefrac{1}{240}+\varepsilon)$-approximation for two buyers with XOS valuations over $m$ items requires $\exp(\Omega(\varepsilon^2 \cdot m))$ communication, 
	 whereas a non-truthful algorithm by Dobzinski and Schapira [SODA 2006] and Feige [2009] is already known to achieve a $\nicefrac{3}{4}$-approximation in $\poly(m)$ communication. 
	
	\smallskip 

We obtain our separation by proving that any {simultaneous} protocol ({not} necessarily truthful) which guarantees a $(\nicefrac{3}{4}-\nicefrac{1}{240}+\varepsilon)$-approximation requires communication $\exp(\Omega(\varepsilon^2 \cdot m))$. The taxation complexity framework of Dobzinski [FOCS 2016] extends this lower bound to all  truthful mechanisms (including interactive truthful mechanisms).

\end{abstract}

\clearpage

\setcounter{tocdepth}{3}
\tableofcontents

\clearpage

\pagenumbering{arabic}
\setcounter{page}{1}

\section{Introduction}\label{sec:intro}
Combinatorial auctions have been at the forefront of Algorithmic Game Theory since the field's inception, owing both to their rich algorithmic theory and  their economic relevance. In a combinatorial auction, there are $n$ bidders, and a seller selling a set $M$ of $m$ items. Each bidder~$i$ has a value for all possible subsets of the items, given by a valuation function $v_i: \mathbbm{2}^{M} \to \mathbb{R}_{+}$. The seller's goal is to find a partition of the $M$ items into disjoint sets $S_1, \cdots, S_n$ such that the \emph{welfare}, $\sum_{i \in [n]} v_i(S_i)$, is maximized. 

The seller faces two challenges in solving this problem. First, the seller must communicate efficiently with the bidders to find a good allocation. Specifically, the seller hopes to use $\poly(n,m)$ total bits of communication, even though each bidder's full valuation function in principle requires (at least) $2^m$ bits to describe. Second, the seller must accommodate the bidders' own incentives. Specifically, the seller desires a protocol that each bidder is incentivized to follow---such protocols are called \emph{truthful}.

The main question we study in this paper is the following: Are there settings where non-truthful algorithms are strictly more powerful than truthful mechanisms? More specifically: Is it the case that for all valuation classes $\mathcal{V}$ and all $\alpha$, if a poly-communication algorithm can guarantee an $\alpha$-approximation when all bidders have valuations in $\mathcal{V}$, then a poly-communication truthful mechanism can also guarantee an $\alpha$-approximation when all bidders have valuations in $\mathcal{V}$?

Our main result is the first setting for which the answer is `no', and in fact we show this separation for the well-studied class of XOS (equivalently, fractionally subadditive) valuation functions.\footnote{A valuation function is XOS if it can be written as a maximum of additive functions---see Section~\ref{sec:caprelim} for precise definitoin.} Before detailing our result, we provide some context.

\paragraph{The VCG Mechanism.}  For some valuation classes $\mathcal{V}$, truthful mechanisms are indeed as powerful as non-truthful algorithms, due to the Vickrey-Clarke-Groves mechanism (\cite{Vickrey61,Clarke71,Groves73}). In TCS terminology, the VCG mechanism is a black-box reduction from \emph{exact} welfare maximization with a truthful mechanism to \emph{exact} welfare maximization with a non-truthful algorithm. More specifically, the VCG mechanism is truthful, maximizes welfare exactly, and can be implemented using $n+1$ black-box calls to a non-truthful algorithm which maximizes welfare exactly. 

There are indeed some restricted settings ({\em e.g.}~when $\mathcal{V}$ is the set of additive valuations, or unit-demand valuations, and even up to Gross Substitutes) for which a poly-communication algorithm precisely maximizes welfare, implying that VCG is also poly-communication and precisely maximizes welfare. Still, the cases for which VCG is poly-communication are \emph{very} restrictive, and do not include, {\em e.g.}, submodular\footnote{A valuation function is submodular if $v(S) + v(T) \leq v(S \cap T) + v(S \cup T)$.} valuations, let alone XOS or subadditive.\footnote{A valuation function is subadditive if $v(S) + v(T) \leq v(S \cup T)$.}

If one considers {\em approximate} welfare maximization, then, for general (unrestricted) valuation functions, the best achievable approximation guarantee by a poly-communication algorithm is just $O(1/\sqrt{m})$~\cite{NisanS06}. Due to the strength of this lower bound, poly-communication ``VCG-based'' truthful mechanisms actually suffice to match this guarantee~\cite{Raghavan88,LehmannOS02, LaviS05}. So in these domains too, poly-communication truthful mechanisms are as powerful as poly-communication algorithms. Still, the guarantees achievable without any assumptions are quite weak.

In summary, truthful mechanisms are as powerful as non-truthful algorithms at the extremes. When valuations are heavily restricted, VCG is poly-communication. When valuations are arbitrary, good poly-communication algorithms don't exist. Still, this leaves out the entire intermediate range of valuation classes.

\paragraph{Beyond VCG: Gaps in Relevant Cases.} Consider now this intermediate range of valuations, such as submodular, XOS, or subadditive: these classes are rich enough to contain realistic valuation functions, yet also restrictive enough to admit poly-communication constant-factor approximation algorithms. For these valuation classes, the state of affairs is drastically different. Indeed, there are huge gaps between the best-known poly-communication algorithm (where deterministic, constant-factor approximations are known for all three classes~\cite{DobzinskiS06, Feige09, FeigeV10}) and the best-known poly-communication truthful mechanism (where no randomized constant-factor approximation is known for any class~\cite{Dobzinski07,AssadiS19,AssadiKS20}, and the best deterministic mechanism guarantees only an $\Omega(1/\sqrt{m})$-approximation~\cite{DobzinskiNS10}). Yet despite these huge gaps in the state of affairs, it was previously unknown whether \emph{any} gap (even a small constant factor) exists in any domain! Our main result provides the first such separation:

\begin{result}[Informal]\label{res:main}
	No poly-communication, deterministic truthful mechanism for two bidders with XOS valuations achieves an approximation guarantee better than $\frac{179}{240}=\frac{3}{4}-\frac{1}{240}$, whereas a poly-communication, deterministic non-truthful algorithm guarantees a $\frac{3}{4}$-approximation.
\end{result}

We note that the part of our main result that deals with non-truthful algorithms is well known and due to~\cite{DobzinskiS06,Feige09}. Our contribution is the lower bound for deterministic truthful mechanisms. In fact, our result generalizes to rule out certain randomized mechanisms as well, but we defer the formal statement to \autoref{thm:mainformal}.

\paragraph{Brief Overview of Approach: Simultaneous Communication.} Communication lower bounds which hold for truthful mechanisms \emph{but not algorithms} are notoriously hard to come by. Specifically, only two general approaches are known. The first is to pick a subclass of truthful mechanisms (e.g., VCG-based), and prove lower bounds against these particular mechanisms. Aforementioned prior work successfully provides such bounds, so we now know that VCG-based truthful mechanisms cannot beat an $O(m^{-1/3})$-approximation for submodular (or XOS, subadditive) valuations~\cite{DobzinskiN11, BuchfuhrerDFKMPSSU10, DanielySS15}. While VCG-based mechanisms are surprisingly general~\cite{LaviMN03}, (deterministic) truthful mechanisms exist which are not VCG-based~\cite{DobzinskiN15, KrystaV12, Dobzinski16a, AssadiS19, AssadiKS20}, and these mechanisms indeed achieve better approximation guarantees than the aforementioned lower bounds. In particular, simple posted-price mechanisms are not VCG-based.\footnote{A posted-price mechanism computes prices $p_1,\ldots, p_m$ in poly-time, then visits each buyer one at a time and asks them to purchase their favorite set (the one maximizing $v_i(S) - \sum_{j \in S} p_j$).} 

The only alternative framework was recently proposed in~\cite{Dobzinski16b}, which establishes the following remarkable theorem (stated formally in \autoref{thm:dobzinski}): if there exists a deterministic poly-communication truthful mechanism which achieves an $\alpha$-approximation for two buyers with XOS valuation functions, then there also exists a deterministic poly-communication \emph{simultaneous} algorithm which achieves an $\alpha$-approximation for two buyers with XOS valuation functions (that is, the two bidders each send exactly one message, simultaneously, and then the designer allocates based only on these messages).\footnote{Note that~\cite{Dobzinski16b} has implications beyond XOS, beyond deterministic protocols, and beyond two bidders, but the implications are tricky to formally state and not relevant for this paper.} That is, while the existence of interactive poly-communication algorithms generally does not imply the existence of simultaneous poly-communication algorithms (e.g.~\cite{PapadimitriouS82, DurisZS84, NisanW93, BabaiGKL03, DobzinskiNO14, AlonNRW15, Assadi17}), the additional structure on interactive \emph{truthful mechanisms} does (at least, for two player combinatorial auctions). Following~\cite{Dobzinski16b}, the remaining task was `merely' to establish a separation between the approximation guarantees achievable in poly-communication with simultaneous versus interactive communication.

Initially, it seems tempting to conjecture that better than just a $\nicefrac{1}{2}$-approximation (which for two bidders is trivial---simply ask each bidder for $v_i(M)$ simultaneously and award $M$ to the highest bidder) would be impossible with poly-communication simultaneous algorithms, due to known lower bounds on ``sketching'' valuation functions~\cite{BadanidiyuruDFKNR12}. However, surprising barriers were discovered on this front:~\cite{BravermanMW18} develop a simultaneous, randomized $\nicefrac{3}{4}$-approximation with poly-communication for two buyers with binary-XOS valuations,\footnote{$v(\cdot)$ is binary-XOS if there exists a collection $\mathcal{C}$ of sets and $v(S):= \max_{T \in \mathcal{C}}\{|S \cap T|\}$. Binary-XOS implies XOS.} which is tight even for interactive algorithms with poly-communication. In addition,~\cite{EzraFNTW19} establish that even interactive algorithms with poly-communication cannot beat a $\nicefrac{1}{2}$-approximation for two bidders with subadditive valuations (which is matched by the aforementioned trivial simultaneous protocol, so there cannot possibly be a separation for two subadditive bidders). We prove our main result by establishing a lower bound of $\nicefrac{3}{4}-\nicefrac{1}{240}$ on the  approximation guarantee of any deterministic, simultaneous algorithm for two bidders with binary-XOS valuation functions, thus also providing the first successful instantiation of Dobzinski's framework~\cite{Dobzinski16b}, despite these barriers.

As the main ideas behind our construction require preliminaries and a detailed overview of prior work (especially~\cite{BravermanMW18}), we defer further details of our proof to the technical sections. We conclude with a reminder that our main result is the first separation between approximation guarantees achievable by (deterministic) truthful mechanisms and (deterministic) algorithms with poly-communication, which follows by providing the first separation between approximation guarantees achievable by (deterministic) simultaneous algorithms and (deterministic) interactive algorithms with poly-communication for two bidders, and an application of~\cite{Dobzinski16b}. 

\subsection{Related Work}\label{sec:related}

\paragraph{Communication complexity separations.} As mentioned above, there are no previously-known separations between approximation guarantees provided by poly-communication truthful mechanisms and poly-communication algorithms. However, some partial results are known.

For example, due to works of~\cite{DobzinskiN11,BuchfuhrerDFKMPSSU10, DanielySS15}, we have a separation between poly-communication algorithms and poly-communication ``VCG-based'' truthful mechanisms when the valuation functions are submodular, XOS, or subadditive. While this rules out a large class of potential mechanisms, we have already noted that (variants of) posted-price mechanisms, which are not VCG-based, outperform these lower bounds. Therefore, more general results (like ours) are necessary to consider these mechanisms.

Along similar lines,~\cite{DobzinskiN15} establishes that a separation  between polylogarithmic communication algorithms and polylogarithmic communication ``scalable'' truthful mechanisms, for the special case of multi-unit auctions (where all items are identical, so a buyer's valuation is fully specified by $m$ numbers). Scalability is not a particularly restrictive definition, although the result is still quite specialized because of its focus on multi-unit auctions (where the entire valuation function can be communicated with $\poly(m)$ bits). 

\paragraph{Other complexity measures.} We conclude with a brief overview of the line of work on {\em computational complexity} of combinatorial auctions. In this setting, the resource of interest is the {\em running-time} of the bidders and the seller during the mechanism. The VCG mechanism again shows that poly-time truthful mechanisms are as powerful as poly-time algorithms in the restricted settings where precise welfare maximization is poly-time tractable. 

Interestingly, welfare-maximization is already inapproximable in poly-time better than $\Theta(m^{-1/2})$ for XOS or subadditive valuations (unless $\bm{\mathrm{P}} = \bm{\mathrm{NP}}$), and again a VCG-based truthful mechanism matches this guarantee~\cite{DobzinskiNS10}. Note the distinction to the communication model, where XOS and subadditive valuations admit a poly-communication constant-factor approximation.

In the computational model, submodular valuations are the sweet spot where constant-factor poly-time approximations exist (but not poly-time exact solutions). Specifically, there is a poly-time $(1-1/e)$-approximation~\cite{Vondrak08}, which is optimal assuming $\bm{\mathrm{P}} = \bm{\mathrm{NP}}$~\cite{MirrokniSV08}. Yet, no (randomized) poly-time truthful mechanism can guarantee a $m^{-1/2+\varepsilon}$-approximation for any $\varepsilon > 0$ (unless $\bm{\mathrm{NP}} \subseteq \bm{\mathrm{RP}}$). Details about this separation can be found in the line of work due to~\cite{Vondrak08, MirrokniSV08, Dobzinski11, DughmiV11, DobzinskiV12b, DobzinskiV12a, DobzinskiV16}.

While this line of works in the computational model is quite impressive, we briefly note one major aspect which is better captured by the communication model. Some algorithms/mechanisms are poly-time as long as the bidders can implement \emph{demand queries}.\footnote{A demand query takes as input a price vector $\vec{p}$ and output the set $\arg\max_S\{v(S) - \sum_{i \in S} p_i\}$.} This includes the $(1-1/e)$-approximation algorithm for XOS valuations~\cite{DobzinskiS06}, the $\nicefrac{1}{2}$-approximation algorithm for subadditive valuations~\cite{Feige09}, and the $O((\log\log m)^{-3})$-approximation truthful mechanism for XOS valuations~\cite{AssadiS19} as well as subadditive valuations~\cite{AssadiKS20}. However, none of these algorithms/mechanisms are ``truly poly-time'' (unless $\bm{\mathrm{NP}} \subseteq \bm{\mathrm{RP}}$), as demand-queries are \bm{\mathrm{NP}}-hard even for submodular valuations.

This means that computational lower bounds \emph{do not} rule out poly-time approximations with demand-queries, and indeed the aforementioned algorithms/mechanisms outperform known computational lower bounds. Put another way, the computational model declares these algorithms/mechanisms to be not poly-time \emph{only because the computational model assumes that bidders cannot choose a set to purchase from a simple pricing scheme in poly-time}. Communication lower bounds do not face this issue, as bidders can clearly \emph{state} the set they wish to purchase with $m$ bits. Along these lines, our results are also the first lower bounds separating what is achievable for algorithms and truthful mechanisms with polynomially-many demand queries. We refer the reader to~\cite{CaiTW20} or~\cite{BravermanMW18} for a deeper comparison of the two models.

\subsection{Roadmap}
In Section~\ref{sec:caprelim}, we provide the minimum preliminaries necessary to state our main result, and to follow with a detailed proof overview in Section~\ref{sec:sketch}. Afterwards, we provide thorough preliminaries necessary for our proofs in Section~\ref{sec:prelim2}, followed by a complete description of our construction in Section~\ref{sec:dist}, and its analysis in Section~\ref{sec:lower}. Appendix~\ref{sec:info} contains the basic information theory tools we use in this paper. 

\section{Problem Statement and Main Result} \label{sec:caprelim}
We first formally define the setting of two player combinatorial auctions.
Let $m > 0$ denote the number of items, and $\mathcal{V}$ be a non-empty set of functions from $\mathbbm{2}^{[m]}$ to  $\mathbb{R}$. A deterministic protocol $\Pi$ for the $m$-item, $\mathcal{V}$-combinatorial auction problem with two bidders is formally specified by the following five functions:
\begin{itemize}
\item $f^\alice$ determines Alice's behavior in the protocol. Specifically, $f^\alice$ takes as input Alice's valuation function $v^\alice \in \mathcal{V}$, and the transcript $ \sigma^{\alice} \in (\{0,1\}^*)^*$ of communication with the Seller she has seen so far, and decides which message (in $\{0,1\}^*$) to next send the Seller. Alice communicates exclusively with the Seller (and not directly with Bob). 
\item $f^\bob$ determines Bob's behavior in the protocol. Similarly, $f^\bob$ takes as input Bob's valuation function $v^\bob \in \mathcal{V}$, and the transcript $\sigma^{\bob} \in (\{0,1\}^*)^*$ of communication with the Seller he has seen so far, and decides which message (in $\{0,1\}^*$) to next send the Seller. Bob communicates exclusively with the Seller (and not directly with Alice).
\item $f^\seller$ determines the Seller's behavior in the protocol. $f^\seller$ takes as input the transcripts $\sigma^{\alice \to \seller},\sigma^{\bob \to \seller} \in (\{0,1\}^*)^*$ it has seen so far, and selects a pair $\{0,1\}^*\times \{0,1\}^* \cup \{(\bot,\bot)\}$ to send. When $\bot$ is sent to both parties, the communication ends.

\item $\alloc$ determines how to allocate the items, once the communication has concluded. Specifically, $\alloc$ takes as input the entirety of Alice's and Bob's communication with the auctioneer (which is in $(\{0,1\}^*)^* \times (\{0,1\}^*)^*$) and selects a pair of sets $(O^\alice,O^\bob) \in \mathbbm{2}^{[m]} \times \mathbbm{2}^{[m]}$, satisfying $O^\alice \cap O^\bob = \emptyset$, to award Alice and Bob respectively. 

\item $\price$ determines how to charge prices, once the communication has concluded. Similarly, $\price$ takes as input the entirety of Alice's and Bob's communication with the Seller (which is in $(\{0,1\}^*)^* \times (\{0,1\}^*)^*$) and selects a pair of prices $(p^\alice,p^\bob) \in \mathbb{R}\times \mathbb{R}$ to charge Alice and Bob, respectively.
\end{itemize}
 
Observe that the functions $f^{\seller}, \alloc, \price$ output a pair (a message/set/price for Alice, and another for Bob). We shall use $f^{\seller \to \alice}$ (respectively, $f^{\seller \to \bob}$) to denote the function that outputs only the message to send to Alice (respectively, the message to send to Bob). We define the functions $\alloc^{\alice}, \alloc^{\bob}, \price^{\alice}, \price^{\bob}$ analogously. We also define a randomized protocol to be a distribution over deterministic protocols.

\paragraph{Execution of a Protocol.} A deterministic, $m$-item, $\mathcal{V}$-combinatorial auction $\Pi = (f^{\alice}, f^{\bob}, f^{\seller}, \alloc, \price)$ takes place as follows: At the beginning of the protocol, the Seller has $m$ items for sale and Alice and Bob have functions $v^{\alice} \in \mathcal{V}$ and $v^{\bob} \in \mathcal{V}$ respectively as input. The protocol takes place in multiple rounds, where before round $i$, for $i > 0$, it holds that Alice has received a transcript $\sigma^{\alice}_{<i} \in \left(\{0,1\}^*\right)^{i-1}$ from the Seller, Bob has received a transcript $\sigma^{\bob}_{<i} \in \left(\{0,1\}^*\right)^{i-1}$ from the Seller, and the Seller has received transcripts  $\sigma^{\alice \to \seller}_{<i}, \sigma^{\bob \to \seller}_{<i} \in \left(\{0,1\}^*\right)^{i-1}$ from Alice, Bob respectively.

In round $i$, Alice and Bob send messages $\sigma^{\alice \to \seller}_{i} = f^{\alice}(v^{\alice}, \sigma^{\alice}_{<i})$ and $\sigma^{\bob \to \seller}_{i} = f^{\bob}(v^{\bob}, \sigma^{\bob}_{<i})$ to the Seller respectively. The Seller appends these to the transcripts $\sigma^{\alice \to \seller}_{<i}, \sigma^{\bob \to \seller}_{<i}$ to get transcripts $\sigma^{\alice \to \seller}_{\leq i}, \sigma^{\bob \to \seller}_{\leq i} \in \left(\{0,1\}^*\right)^{i}$. Thereafter, the seller sends a message  $\sigma^{\alice}_{i} = f^{\seller \to \alice}(\sigma^{\alice \to \seller}_{\leq i}, \sigma^{\bob \to \seller}_{\leq i})$ to Alice and a message $\sigma^{\bob}_{i} = f^{\seller \to \bob}(\sigma^{\alice \to \seller}_{\leq i}, \sigma^{\bob \to \seller}_{\leq i})$ to Bob.

If $(\sigma^{\alice}_{i}, \sigma^{\bob}_{i}) \neq (\bot, \bot)$, then  Alice (resp. Bob) append $\sigma^{\alice}_{i}$ to $\sigma^{\alice}_{<i}$ (resp. $\sigma^{\bob}_{i}$ to $\sigma^{\bob}_{<i}$) to get transcript $\sigma^{\alice}_{\leq i}$ (resp. $\sigma^{\bob}_{\leq i}$) and continue round $i+1$ of the protocol. On the other hand, if $(\sigma^{\alice}_{i}, \sigma^{\bob}_{i}) = (\bot, \bot)$, then the protocol {\em terminates} after round $i$ and no further communication takes place. The Seller outputs an allocation $(O^{\alice}, O^{\bob}) = \alloc(\sigma^{\alice \to \seller}_{\leq i}, \sigma^{\bob \to \seller}_{\leq i})$, and prices $(p^{\alice}, p^{\bob}) = \price(\sigma^{\alice \to \seller}_{\leq i}, \sigma^{\bob \to \seller}_{\leq i})$.

Observe that, if $\Pi$ is deterministic, then, the values of $(O^{\alice}, O^{\bob}) $ and $(p^{\alice}, p^{\bob})$ are completely determined by $\Pi$ and the inputs $v^{\alice}, v^{\bob}$ to Alice and Bob respectively. We sometimes denote these values by $ (O^{\alice}, O^{\bob})  = \alloc_{\Pi}(v^{\alice}, v^{\bob})$ and $ (p^{\alice}, p^{\bob})  = \price_{\Pi}(v^{\alice}, v^{\bob})$. We will also use the shorthand $O^{\alice}  = \alloc^{\alice}_{\Pi}(v^{\alice}, v^{\bob})$, {\em etc.}

\paragraph{Properties of a Protocol.} We consider the following parameters of a protocol:
\begin{itemize}
\item {\bf Rounds:} For a deterministic protocol $\Pi$, and $ v^{\alice}, v^{\bob} \in \mathcal{V}$, define $R_{\Pi}(v^{\alice}, v^{\bob}) = R$ if the execution of $\Pi$ when Alice and Bob have inputs $v^{\alice}, v^{\bob}$ respectively terminates after round $R$. If the execution does not terminate at all, then we define $R_{\Pi}(v^{\alice}, v^{\bob}) =  \infty$.

We say that $\Pi$ has $R$ rounds if, for all $ v^{\alice}, v^{\bob} \in \mathcal{V}$, we have $R_{\Pi}(v^{\alice}, v^{\bob}) = R$. A randomized protocol has $R$ rounds if all the deterministic protocols in its support have $R$ rounds. If a deterministic or randomized protocol has exactly $1$ round, then, we say that the protocol is {\em simultaneous}.

To emphasize, in a simultaneous protocol, Alice and Bob each send exactly one message. The Seller does not send any messages. Then, an allocation is determined only as a function of these messages. 

\item {\bf Communication complexity:} For a deterministic protocol $\Pi$, and $ v^{\alice}, v^{\bob} \in \mathcal{V}$, we define $\cc_{\Pi}(v^{\alice}, v^{\bob}) = \infty$ if $R_{\Pi}(v^{\alice}, v^{\bob}) = \infty$. On the other hand, if $R_{\Pi}(v^{\alice}, v^{\bob}) = R < \infty$, then we define
\[
\cc_{\Pi}(v^{\alice}, v^{\bob}) =
\sum_{i \leq R} \len(\sigma^{\alice \to \seller}_{i}) + \len(\sigma^{\bob \to \seller}_{i}) + \sum_{i < R} \len(\sigma^{\alice}_{i}) + \len(\sigma^{\bob}_{i}) 
.\] 
In the above equation, the values $\sigma^{\alice \to \seller}_{i}$, $\sigma^{\bob \to \seller}_{i}$, {\em etc.} denote the corresponding values in an execution of $\Pi$ when Alice has input $v^{\alice}$ and Bob has input $v^{\bob}$. These values are well defined as $\Pi$ is deterministic.

We define $\cc(\Pi) =  \max_{v^{\alice}, v^{\bob} \in \mathcal{V}} \cc_{\Pi}(v^{\alice}, v^{\bob})$. Finally we define $\cc(\Pi')$, for a randomized protocol $\Pi'$  to be the largest value of $\cc(\Pi)$ for all deterministic protocols $\Pi$ in its support.

\item {\bf Truthfulness:} We say that a deterministic protocol $\Pi$ is truthful if for all $v^{\alice}, v^{\bob}, v' \in \mathcal{V}$, following the protocol is an \emph{ex-post Nash}. Formally:
\begin{align*}
v^{\alice}(\alloc^{\alice}_{\Pi}(v^{\alice}, v^{\bob})) - \price^{\alice}_{\Pi}(v^{\alice}, v^{\bob}) &\geq v^{\alice}(\alloc^{\alice}_{\Pi}(v', v^{\bob})) - \price^{\alice}_{\Pi}(v', v^{\bob})\\
v^{\bob}(\alloc^{\bob}_{\Pi}(v^{\alice}, v^{\bob})) - \price^{\bob}_{\Pi}(v^{\alice}, v^{\bob}) &\geq v^{\bob}(\alloc^{\bob}_{\Pi}(v^{\alice}, v')) - \price^{\bob}_{\Pi}(v^{\alice}, v')
\end{align*}

We say that a randomized protocol is \emph{universally truthful} if all the deterministic mechanism in its support are truthful. To clearly emphasize the distinction between protocols which are truthful and not truthful, we will often refer to a truthful protocol as a (truthful) \emph{mechanism}, and one which is not necessarily truthful as an \emph{algorithm}.

\item {\bf Approximation guarantee:} For $m, \mathcal{V}$ as above and $v^{\alice}, v^{\bob} \in \mathcal{V}$, define the function $\opt(v^{\alice}, v^{\bob}) = \max_{S^{\alice}, S^{\bob} \subseteq [m]: S^{\alice} \cap  S^{\bob} = \emptyset}v^{\alice}(S^{\alice}) + v^{\bob}(S^{\bob})$. Let $\nu$ be a distribution over pairs drawn from $\mathcal{V}$ and $\alpha, p > 0$. We say that a deterministic mechanism $\Pi$ is $\alpha$-approximate over $\nu$ with probability $p$ if we have 
\[
\Pr_{(v^{\alice}, v^{\bob}) \sim \nu}\left(v^{\alice}(\alloc^{\alice}_{\Pi}(v^{\alice}, v^{\bob})) + v^{\bob}(\alloc^{\bob}_{\Pi}(v^{\alice}, v^{\bob})) > \alpha \cdot \opt(v^{\alice}, v^{\bob}) \right) \geq p.
\]

We further say that a randomized mechanism $\Pi'$ is $\alpha$-approximate with probability $p$ if for all $v^{\alice}, v^{\bob} \in \mathcal{V}$, we have:
\[
\Pr_{\Pi} \left(v^{\alice}(\alloc^{\alice}_{\Pi}(v^{\alice}, v^{\bob})) + v^{\bob}(\alloc^{\bob}_{\Pi}(v^{\alice}, v^{\bob})) > \alpha \cdot \opt(v^{\alice}, v^{\bob}) \right) \geq p,
\]
where the probability is over all deterministic mechanisms $\Pi$ in the support of $\Pi'$.

\end{itemize}

\subsection{Formal Statement of Our Main Result} \label{sec:mainformal}

We now formalize our main result. For $m > 0$, let $\bxos_m$ be the class of all Binary-XOS functions on $m$ items. That is, $\bxos_m$ denotes the set of all $v : \mathbbm{2}^{[m]} \to \mathbb{R}$ such that there exists a collection $\mathcal{C} \subseteq \mathbbm{2}^{[m]}$, such that for all $S \in \mathbbm{2}^{[m]}$, $v(S) =  \max_{C \in \mathcal{C}} \{\lvert{S \cap C}\rvert\}$. Define also $\xos_m \supseteq \bxos_m$ to be the class of all XOS functions on $m$ items. That is, $\xos_m$ denotes the set of all $v: \mathbbm{2}^{[m]} \to \mathbb{R}$ such that there exists a collection $\mathcal{C} \subseteq \mathbb{R}_+^{m}$, such that for all $S \in \mathbbm{2}^{[m]}$, we have $v(S) =  \max_{c \in \mathcal{C}} \{\sum_{i \in S} c_i\}$.

\begin{theorem}[Main Result]
\label{thm:mainformal} 
There exists a constant $\beta > 0$ such that for all $\varepsilon > 0$, there is an $m_0 > 0$ satisfying the following: For all  $m > m_0$, any randomized, $m$-item, $\xos_m$-combinatorial auction $\Pi$ with two bidders and one seller that is universally truthful and $\left(\nicefrac{3}{4} - \nicefrac{1}{240} + \varepsilon \right)$-approximate with probability $\nicefrac{1}{2} + \exp( - \beta \varepsilon^2 \cdot m)$ satisfies
\[
\cc(\Pi) \geq \exp(\beta \varepsilon^2 \cdot m).
\]
\end{theorem}
Note, of course, that deterministic protocols are a special case of randomized protocols, so Theorem~\ref{thm:mainformal} also applies to deterministic mechanisms. Combining this with the deterministic $\nicefrac{3}{4}$-approximation for $\xos_m$ which uses only $\poly(m)$ communication~\cite{DobzinskiS06,Feige09} separates the achievable guarantees of deterministic truthful mechanisms and deterministic algorithms with poly-communication.

Our proof of~\autoref{thm:mainformal} makes use of the Taxation Complexity framework developed by~\cite{Dobzinski16b}. This framework is very rich, and has implications beyond XOS valuations, and beyond two-player auctions. We state below the only case of the framework necessary for our main results, and refer the reader to~\cite{Dobzinski16b} for the full framework.

\begin{theorem}[\cite{Dobzinski16b}]
\label{thm:dobzinski}
There exists a polynomial $P(\cdot)$ such that for all $m,p, \alpha > 0$ and all randomized, $m$-item, $\xos_m$-combinatorial auction $\Pi$ with two bidders and one seller that are universally truthful and $\alpha$-approximate with probability $p$, there is a randomized, $m$-item, $\xos_m$-combinatorial auction $\Pi'$ with two bidders and one seller that is simultaneous and $\alpha$-approximate with probability $p$, and satisfies $\cc(\Pi') \leq P(\max(\cc(\Pi), m))$.
\end{theorem}

Theorem~\ref{thm:dobzinski} provides a poly-communication reduction from simultaneous combinatorial auctions to truthful combinatorial auctions. Our main technical result is a lower bound on the simultaneous communication necessary for a randomized protocol that is $\left(\nicefrac{3}{4} - \nicefrac{1}{240} + \varepsilon \right)$-approximate with probability $\nicefrac{1}{2} + \exp( - \beta \varepsilon^2 \cdot m)$.

\begin{theorem}
\label{thm:mainred} 
For all $\varepsilon > 0$, and all  $m > \frac{10^{10}}{\varepsilon^2}$, any randomized, $m$-item, $\bxos_m$-combinatorial auction $\Pi$ with two bidders and one seller that is simultaneous and $\left(\frac{3}{4} - \frac{1}{240} + \varepsilon \right)$-approximate with probability $\frac{1}{2} + \exp\left( - \frac{ \varepsilon^2 m }{500}\right)$ satisfies
\[
\cc(\Pi) \geq \exp\left(\frac{ \varepsilon^2 m }{500}\right).
\]

\end{theorem}
We briefly compare~\Cref{thm:mainred} to Theorem~1.1 of~\cite{BravermanMW18}. Theorem~1.1 of~\cite{BravermanMW18} gives a randomized, poly-communication simultaneous algorithm which gets a $\nicefrac{3}{4}$-approximation \emph{in expectation}.~\Cref{thm:mainred} rules out randomized, poly-communication simultaneous algorithms which achieve a $\nicefrac{3}{4}$-approximation \emph{with probability slightly more than $1/2$} (including deterministic algorithms). 

For the sake of completeness, we prove~\autoref{thm:mainformal} assuming~\autoref{thm:dobzinski} and~\autoref{thm:mainred} in Appendix~\ref{app:proofs}. The remainder of the paper is devoted to proving~\autoref{thm:mainred}. By Yao's minimax principle, in order to a lower bound $\cc(\Pi)$ for randomized $m$-item simultaneous mechanisms $\Pi$ that are $\alpha$-approximate with probability $p$ (for some $m, \alpha, p$), it is sufficient to show a distribution $\nu$ over pairs of functions in $\bxos_m$, such that all deterministic simultaneous mechanisms $\Pi'$ that are $\alpha$-approximate over $\nu$ with probability $p$ have large $\cc(\Pi')$. We construct $\nu$ in~\Cref{sec:dist} and analyze it in \Cref{sec:lower}. Before this, we give a detailed sketch of our construction, and the key aspects that drive it.

\section{Detailed Proof Sketch}\label{sec:sketch}
In this section, we gradually build various aspects of our main construction and highlight the roles they play. All valuation functions for the rest of the paper will be BXOS. Recall that each Binary-XOS valuation $v$ has an associated set $\mathcal{C}$ of \emph{clauses}, such that $v(S):=\max_{T \in \mathcal{C}}\{\lvert S \cap T \rvert\}$. We shall sometimes refer to $v$ simply by its set of clauses.

As mentioned previously, our work builds off a prior construction of~\cite{BravermanMW18}, which we first describe in detail.

\subsection{The~\cite{BravermanMW18} Construction}\label{sec:bmwformal} 
\cite{BravermanMW18} also studies $\bxos$ combinatorial auctions. Their result which serves as our starting point is a lower bound on the communication required to \emph{determine} the value of the optimal achievable welfare up to a factor of $\nicefrac{3}{4}$. Importantly, though, observe that for \emph{simultaneous} protocols, hardness for the decision problem \emph{does not imply} hardness for finding an approximately-optimal allocation (and hardness for the decision problem has no implications in Dobzinski's framework). Indeed, deciding the optimal achievable welfare in the~\cite{BravermanMW18} construction better than a $(\nicefrac{3}{4}-\nicefrac{1}{108})$-approximation requires exponential communication, yet an allocation guaranteeing a $\nicefrac{3}{4}$-approximation can be found with polynomial communication! We elaborate on this after presenting their construction.\footnote{To get quick intuition for how this can be ever possible, consider the trivial reduction establishing that allocation is at least as hard as decision: first, solve the allocation problem; then, ask Alice and Bob to output their value for the allocation chosen, and solve the decision problem. This reduction requires an extra round for Alice and Bob to evaluate the solution, and so it cannot be applied simultaneously. One interpretation of~\cite{BravermanMW18} is that this extra round is necessary.}

In the construction of \cite{BravermanMW18}, the valuation functions of Alice and Bob are BXOS with exponentially many {\em regular} clauses, and may or may not include one {\em special} clause. The regular clauses are constructed so that the union of a regular clause of Alice and a regular clause of Bob has size $<\nicefrac{3m}{4}$ (and therefore, the maximum possible welfare of any allocation is $<\nicefrac{3m}{4}$ as well) while the union of a special clause of Alice and a special clause of Bob has size $m$ (and therefore the optimal allocation has welfare $m$). This means that determining the optimal welfare up to a factor of $\nicefrac{3}{4}$ (or in fact, any constant better than $\nicefrac{20}{27}$) amounts to determining whether or not Alice and Bob have special clauses.

However, in the~\cite{BravermanMW18} construction, the special clauses of Alice and Bob are indistinguishable from the regular clauses. Intuitively, determining whether or not one of their exponentially many clauses is special \emph{with a simultaneous protocol} then requires exponential communication (and this is true). We now detail the~\cite{BravermanMW18} construction.

\begin{figure}[t]

\begin{center}

\begin{tikzpicture}

\node[draw, rectangle, minimum width=20pt, minimum height=20pt, inner sep=0pt] (11){\scriptsize 1}; 
\foreach \x in {2,...,6}
{
	\pgfmathtruncatemacro{\prev}{\x - 1}
	\pgfmathtruncatemacro{\name}{1\prev}
	\node[draw, rectangle, minimum width=20pt, minimum height=20pt, inner sep=0pt] (1\x) [right=5pt of \name]{\scriptsize \x}; 
}

\node[draw, rectangle, minimum width=20pt, minimum height=20pt, inner sep=0pt] (21)[below=0.4cm of 11]{$\checkmark$}; 
\foreach \x in {2,...,6}
{
	\pgfmathtruncatemacro{\prev}{\x - 1}
	\pgfmathtruncatemacro{\name}{2\prev}
	
	\ifthenelse{\x=1 \OR \x=3 \OR \x=4}{
	\node[draw, rectangle, minimum width=20pt, minimum height=20pt, inner sep=0pt] (2\x) [right=5pt of \name]{$\checkmark$}; 
	}{
	\node[draw, rectangle, minimum width=20pt, minimum height=20pt, inner sep=0pt] (2\x) [right=5pt of \name]{};
	} 
}
\node (S1) [left=0.25cm of 21]{$S:$};

\node[draw, rectangle, minimum width=20pt, minimum height=20pt, inner sep=0pt] (31)[below=0.05cm of 21]{}; 
\foreach \x in {2,...,6}
{
	\pgfmathtruncatemacro{\prev}{\x - 1}
	\pgfmathtruncatemacro{\name}{3\prev}
	
	\ifthenelse{\x=2 \OR \x=3 \OR \x=4}{
	\node[draw, rectangle, minimum width=20pt, minimum height=20pt, inner sep=0pt] (3\x) [right=5pt of \name]{$\checkmark$}; 
	}{
	\node[draw, rectangle, minimum width=20pt, minimum height=20pt, inner sep=0pt] (3\x) [right=5pt of \name]{};
	}
}
\node (T1) [left=0.25cm of 31]{$T:$};

\node[draw, rectangle, minimum width=20pt, minimum height=20pt, inner sep=0pt] (41)[below=0.4cm of 31]{$\checkmark$}; 
\foreach \x in {2,...,6}
{
	\pgfmathtruncatemacro{\prev}{\x - 1}
	\pgfmathtruncatemacro{\name}{4\prev}
	
	\ifthenelse{\x=1 \OR \x=3 \OR \x=5}{
	\node[draw, rectangle, minimum width=20pt, minimum height=20pt, inner sep=0pt] (4\x) [right=5pt of \name]{$\checkmark$}; 
	}{
	\node[draw, rectangle, minimum width=20pt, minimum height=20pt, inner sep=0pt] (4\x) [right=5pt of \name]{};
	} 
}
\node (S2) [left=0.25cm of 41]{$A:$};

\node[draw, rectangle, minimum width=20pt, minimum height=20pt, inner sep=0pt] (51)[below=0.05cm of 41]{}; 
\foreach \x in {2,...,6}
{
	\pgfmathtruncatemacro{\prev}{\x - 1}
	\pgfmathtruncatemacro{\name}{5\prev}
	
	\ifthenelse{\x=2 \OR \x=4 \OR \x=6}{
	\node[draw, rectangle, minimum width=20pt, minimum height=20pt, inner sep=0pt] (5\x) [right=5pt of \name]{$\checkmark$}; 
	}{
	\node[draw, rectangle, minimum width=20pt, minimum height=20pt, inner sep=0pt] (5\x) [right=5pt of \name]{};
	} 
}
\node (T2) [left=0.25cm of 51]{$B:$};

\end{tikzpicture}
\end{center}

\caption{The construction of \cite{BravermanMW18}. Each of the numbers $1$ to $6$ represents a group of~$\frac{m}{6}$ items.}
\label{fig:bmw}
\end{figure}

\subsubsection{The structure of the clauses in \cite{BravermanMW18}} 

\paragraph{Step One: Select a Basis.} For the~\cite{BravermanMW18} construction, a \emph{basis} is a pair of sets $(S,T)$ such that $|S| = |T|= m/2$, and also $|S \cap T| = m/3$. In the~\cite{BravermanMW18} construction, a basis $(S,T)$ is sampled uniformly at random from all possible bases. Alice knows $S$ and Bob knows $T$ (Alice does not know $T$, but has a Bayesian posterior conditioned on $S$ and the fact that $(S,T)$ is a uniformly random basis). We provide an illustration of one possible basis in~\autoref{fig:bmw} where each of the six blocks in a row represents a group of $\nicefrac{m}{6}$ items.  

\paragraph{Step Two: Draw Regular Clauses.} Alice's regular clauses are constructed by uniformly sampling sets of size $\nicefrac{m}{2}$ that have intersection exactly $\nicefrac{m}{3}$ with~$S$, and Bob's regular clauses are constructed by uniformly sampling sets of size $\nicefrac{m}{2}$ that have intersection exactly $\nicefrac{m}{3}$ with~$T$. Constructing the regular clauses this way satisfies the following first key property: \emph{The union of a regular clause of Alice and a regular clause of Bob has size strictly less than $\nicefrac{3m}{4}$} (in fact, at most $\nicefrac{20m}{27} + \varepsilon m$ except with exponentially small probability). 

We briefly explain why (it is $<\nicefrac{3m}{4}$). As all regular clauses have size $\nicefrac{m}{2}$, it is equivalent to describe why the intersection of a regular clause of Alice and a regular clause of Bob has size strictly more than $\nicefrac{m}{4}$. Intuitively, this is because each regular clause of Alice intersects $S$ more than random, while each regular clause of Bob intersects $T$ more than random, and $S$ and $T$ intersect more than random. Put another way, if the basis $(S,T)$ instead satisfied $|S \cap T| = \nicefrac{m}{4}$, the expected size of the intersection of two {\em independently} random sets of size $\frac{m}{2}$, then, as the regular clauses of Alice and Bob are chosen independently of each other, they will also behave like independently chosen random sets and have an intersection of size $\nicefrac{m}{4}$ in expectation. In actuality, the basis $(S,T)$ has intersection of size $\nicefrac{m}{3}$, more than the expected size of the intersection of two {\em independently} random sets of size $\nicefrac{m}{2}$. Thus, the regular clauses of Alice and Bob also intersect more than random sets, {\em i.e.}, in more than $\nicefrac{m}{4}$ places. 

Importantly, observe that if we were to curtail the construction here, that the optimal welfare would be $< \nicefrac{3m}{4}$.

\paragraph{Step Three: Special Clauses.} The second key property of this construction is that we can `hide' a special clause inside the exponentially many regular clauses sampled by Alice and Bob.

To see an illustration of how a special clause is hidden amongst the regular clauses, observe the rows corresponding to the special clauses $A$ and $B$ in \autoref{fig:bmw}. The special clauses for Alice and Bob are disjoint, and their union is of size $m$. Additionally, note that $A$ intersects $S$ in $\frac{m}{3}$ places and similarly $B$ intersects $T$ in $\frac{m}{3}$ places, just like all the regular clauses. As the size of their intersections with $S$ and $T$ (respectively) are the same, Alice and Bob cannot tell the special clauses (if they are present) apart from the regular clauses. 

Importantly, observe that we can now either add or not add a pair of special clauses to their input. If we do, then the optimal achievable welfare is now $m$. If we don't, it remains $< \nicefrac{3m}{4}$. So for Alice and Bob to \emph{simultaneously} decide whether they have a special clause or not, they must somehow send information about each of their exponentially many clauses, which requires exponential communication.

\paragraph{Two Observations.} We briefly make two observations about the~\cite{BravermanMW18} construction (without proof). First, their lower bound holds only for simultaneous protocols. Indeed, Alice and Bob could first communicate $S$ and $T$ to each other in round one, and then they could declare in round two whether they have a special clause or not. In addition, observe that if we simply award to Alice the items corresponding to a uniformly random clause, this allocation achieves a $> \nicefrac{3m}{4}$-approximation with high probability! We refer the reader to~\cite{BravermanMW18} for these calculations, but note that the main idea is that \emph{Bob can have high welfare for a set because of his special clause, without communicating to the Seller that a special clause exists.} So if we award Alice a uniformly random clause, if Bob happens to have a special clause, then his welfare is at least $m/4$ (and therefore the achievd welfare is at least $3m/4$, good enough for a $\nicefrac{3}{4}$-approximation). If Bob doesn't have a special clause, then the resulting welfare is nearly-optimal. But observe that this approximation is guaranteed \emph{without needing to learn whether Bob has a special clause or not}.

This latter phenomenon is not just an artifact of precise choices in the~\cite{BravermanMW18} construction, but a genuine barrier. For example,~\cite{BravermanMW18} also designs a randomized, poly-communication simultaneous algorithm that achieves a $\nicefrac{3}{4}$-approximation in expectation. Of course, this algorithm is not deterministic, nor does it guarantee a $\nicefrac{3}{4}$-approximation with good probability (see Theorem~\ref{thm:mainformal}). But it does help convey that the allocation and decision problems are fundamentally different for simultaneous algorithms.

\subsubsection{A Minor Generalization} 
In the presented construction, we thought of each of the blocks from $1$ to $6$  in \autoref{fig:bmw} as representing a group of $\nicefrac{m}{6}$ items. However, the exact same arguments (with numerically-different calculations) would also apply to any construction where blocks $1$ and $2$ represented $u$ items, and blocks $3$ through $6$ represented $v$ items (for any $u,v$).

With these additional parameters, it turns out (we omit the calculations), that  the size of the  intersection of a regular clause of Alice and a regular clause of Bob  is:
\[ \frac{2v^3 + 2u^2v + 3uv^2}{(u+2v)^3}\cdot m.\]
The expression above is maximized when $u=v$ (as observed in~\cite{BravermanMW18}) but is strictly larger than $\nicefrac{m}{4}$ for all $u,v$ such that $u < 2v$ (to get intuition for the breakpoint: when $u =2v$, $\lvert{S \cap T}\rvert = \nicefrac{m}{4}$, and  $S,T$ behave like independently chosen sets). We will use this idea later in our construction.

\subsection{From the Decision Problem to the Allocation Problem} \label{sec:dectoalloc}

The crucial difference between~\cite{BravermanMW18} and our work is that~\cite{BravermanMW18} show that the problem of `deciding' whether or not the optimal welfare is close to $m$ is hard while we wish to show that the problem of `computing' an allocation with welfare close to the optimal is hard. As~\cite{BravermanMW18} emphasize, these problems are incomparable for simultaneous mechanisms.

Our construction is based on the following approach of going from a lower bound for the  decision problem to a lower bound for the allocation problem: Consider two copies of the~\cite{BravermanMW18} construction on disjoint sets of items, where (a uniformly chosen) one is such that Alice and Bob have the special clauses and the other one is such that Alice and Bob do not have the special clauses. \textbf{Suppose further that the Seller can only allocate items in one of the two copies.}

We claim that the decision lower bound for~\cite{BravermanMW18} implies an allocation lower bound for this artificial problem. Indeed,  the optimal welfare of the copy with the special clauses is much larger than the optimal welfare of the copy without the special clauses (by more than a factor of $\nicefrac{4}{3}$). Thus, any allocation that allocates items in only one of the two copies and gets welfare close to optimal must allocate items in the copy with the special clause. But, this requires the Seller to at least determine which copy has the special clause, which is hard due to~\cite{BravermanMW18}. The catch, of course, is that we needed to assume that the Seller can only allocate items in one of the two copies, so this is not actually an instance of the combinatorial auctions problem.

\paragraph{Cross-terms.} It remains now to transform the system with two copies and a restriction on the Seller to only allocate items in one of the two copies to a standard combinatorial auction. A first approach may be to have two bases $(S^1, T^1)$ and $(S^2, T^2)$ \emph{on the same set of items} and give Alice and Bob regular clauses generated from both the bases together with a special clause from (a uniformly random) one of the bases. 

One would then hope that just like the system described above, computing a good allocation for this system would require the Seller to implicitly determine which basis has special clause come, and maybe we can show that determining this is hard {\em \`{a} la}~\cite{BravermanMW18}.

Unfortunately, this is not actually the case. The reason is that having two bases on the same set of items gives rise to \emph{cross terms}. Specifically,  if we have two bases on the same set of items, then not only do we have to argue about the size of the union of regular clauses from basis $1$ of Alice and regular clauses from basis $1$ of Bob, but we also need to argue about the size of the union of regular clauses from basis $1$ of Alice and regular clauses from basis $2$ of Bob.

These additional unions, which we call the cross-terms, imply that the two bases must necessarily be correlated in order to avoid the issues described in \Cref{sec:bmwformal}. Namely, if the two bases are independent, then $S^1$ and $T^2$ intersect in  $\nicefrac{m}{4}$ places in expectation (like sets of size $\nicefrac{m}{2}$ chosen independently), implying in turn that the size of the union of regular clauses from basis $1$ of Alice and regular clauses from basis $2$ of Bob is $\nicefrac{3m}{4}m$ in expectation. This is too large for our lower bound, as we need the union to be of size strictly less than $\nicefrac{3m}{4}$ in expectation.

But, we do at least have a candidate approach: pick two correlated bases, and hope to find an appropriate correlation so that knowing an allocation which achieves welfare $\nicefrac{3m}{4}$ immediately determines which basis had a special clause.

\begin{figure}[t]
\begin{center}

\begin{tikzpicture}

\node[draw, rectangle, minimum width=20pt, minimum height=20pt, inner sep=0pt] (11){\scriptsize 1}; 
\foreach \x in {2,...,12}
{
	\pgfmathtruncatemacro{\prev}{\x - 1}
	\pgfmathtruncatemacro{\name}{1\prev}
	\ifthenelse{\x=9}{
	\node[draw, rectangle, minimum width=20pt, minimum height=20pt, inner sep=0pt] (1\x) [right= 0.4cm of \name]{\scriptsize \x}; 
	}{
	\node[draw, rectangle, minimum width=20pt, minimum height=20pt, inner sep=0pt] (1\x) [right=5pt of \name]{\scriptsize \x}; 
	}
}

\node[draw, rectangle, minimum width=20pt, minimum height=20pt, inner sep=0pt] (21)[below=0.4cm of 11]{$\checkmark$}; 
\foreach \x in {2,...,12}
{
	\pgfmathtruncatemacro{\prev}{\x - 1}
	\pgfmathtruncatemacro{\name}{2\prev}
	
	\ifthenelse{\x=9}{
	\node[draw, rectangle, minimum width=20pt, minimum height=20pt, inner sep=0pt] (2\x) [right= 0.4cm of \name]{$\checkmark$}; 
	}{
	\ifthenelse{\x=2 \OR \x=3 \OR \x=7 \OR \x=9 \OR \x=10}{
	\node[draw, rectangle, minimum width=20pt, minimum height=20pt, inner sep=0pt] (2\x) [right=5pt of \name]{$\checkmark$}; 
	}{
	\node[draw, rectangle, minimum width=20pt, minimum height=20pt, inner sep=0pt] (2\x) [right=5pt of \name]{};
	} 
	}
}
\node (S1) [left=0.25cm of 21]{$S^1:$};

\node[draw, rectangle, minimum width=20pt, minimum height=20pt, inner sep=0pt] (31)[below=0.05cm of 21]{}; 
\foreach \x in {2,...,12}
{
	\pgfmathtruncatemacro{\prev}{\x - 1}
	\pgfmathtruncatemacro{\name}{3\prev}
	
	\ifthenelse{\x=9}{
	\node[draw, rectangle, minimum width=20pt, minimum height=20pt, inner sep=0pt] (3\x) [right= 0.4cm of \name]{$\checkmark$}; 
	}{
	\ifthenelse{\x=4 \OR \x=7 \OR \x=5 \OR \x=6 \OR \x=9 \OR \x=10 }{
	\node[draw, rectangle, minimum width=20pt, minimum height=20pt, inner sep=0pt] (3\x) [right=5pt of \name]{$\checkmark$}; 
	}{
	\node[draw, rectangle, minimum width=20pt, minimum height=20pt, inner sep=0pt] (3\x) [right=5pt of \name]{};
	} 
	}
}
\node (S2) [left=0.25cm of 31]{$S^2:$};

\node[draw, rectangle, minimum width=20pt, minimum height=20pt, inner sep=0pt] (41)[below=0.05cm of 31]{}; 
\foreach \x in {2,...,12}
{
	\pgfmathtruncatemacro{\prev}{\x - 1}
	\pgfmathtruncatemacro{\name}{4\prev}
	
	\ifthenelse{\x=9}{
	\node[draw, rectangle, minimum width=20pt, minimum height=20pt, inner sep=0pt] (4\x) [right= 0.4cm of \name]{$\checkmark$}; 
	}{
	\ifthenelse{\x=2 \OR \x=3 \OR \x=4 \OR \x=5 \OR \x=9 \OR \x=10}{
	\node[draw, rectangle, minimum width=20pt, minimum height=20pt, inner sep=0pt] (4\x) [right=5pt of \name]{$\checkmark$}; 
	}{
	\node[draw, rectangle, minimum width=20pt, minimum height=20pt, inner sep=0pt] (4\x) [right=5pt of \name]{};
	}
	} 
}
\node (T1) [left=0.25cm of 41]{$T^1:$};

\node[draw, rectangle, minimum width=20pt, minimum height=20pt, inner sep=0pt] (51)[below=0.05cm of 41]{}; 
\foreach \x in {2,...,12}
{
	\pgfmathtruncatemacro{\prev}{\x - 1}
	\pgfmathtruncatemacro{\name}{5\prev}
	
	\ifthenelse{\x=9}{
	\node[draw, rectangle, minimum width=20pt, minimum height=20pt, inner sep=0pt] (5\x) [right= 0.4cm of \name]{$\checkmark$}; 
	}{
	\ifthenelse{\x=2 \OR \x=7 \OR \x=6 \OR \x=8 \OR \x=9 \OR \x=10}{
	\node[draw, rectangle, minimum width=20pt, minimum height=20pt, inner sep=0pt] (5\x) [right=5pt of \name]{$\checkmark$}; 
	}{
	\node[draw, rectangle, minimum width=20pt, minimum height=20pt, inner sep=0pt] (5\x) [right=5pt of \name]{};
	} 
	}
}
\node (T2) [left=0.25cm of 51]{$T^2:$};

\node[draw, rectangle, minimum width=20pt, minimum height=20pt, inner sep=0pt] (61)[below=0.4cm of 51]{$\checkmark$}; 
\foreach \x in {2,...,12}
{
	\pgfmathtruncatemacro{\prev}{\x - 1}
	\pgfmathtruncatemacro{\name}{6\prev}
	
	\ifthenelse{\x=9}{
	\node[draw, rectangle, minimum width=20pt, minimum height=20pt, inner sep=0pt] (6\x) [right= 0.4cm of \name]{$\checkmark$}; 
	}{
	\ifthenelse{\x=3 \OR \x=6 \OR \x=7 \OR \x=9 \OR \x=12}{
	\node[draw, rectangle, minimum width=20pt, minimum height=20pt, inner sep=0pt] (6\x) [right=5pt of \name]{$\checkmark$}; 
	}{
	\node[draw, rectangle, minimum width=20pt, minimum height=20pt, inner sep=0pt] (6\x) [right=5pt of \name]{};
	} 
	}
}
\node (S1) [left=0.25cm of 61]{$A^1:$};

\node[draw, rectangle, minimum width=20pt, minimum height=20pt, inner sep=0pt] (71)[below=0.05cm of 61]{$\checkmark$}; 
\foreach \x in {2,...,12}
{
	\pgfmathtruncatemacro{\prev}{\x - 1}
	\pgfmathtruncatemacro{\name}{7\prev}
	
	\ifthenelse{\x=9}{
	\node[draw, rectangle, minimum width=20pt, minimum height=20pt, inner sep=0pt] (7\x) [right= 0.4cm of \name]{}; 
	}{
	\ifthenelse{\x=4 \OR \x=5 \OR \x=7 \OR \x=10 \OR \x=11}{
	\node[draw, rectangle, minimum width=20pt, minimum height=20pt, inner sep=0pt] (7\x) [right=5pt of \name]{$\checkmark$}; 
	}{
	\node[draw, rectangle, minimum width=20pt, minimum height=20pt, inner sep=0pt] (7\x) [right=5pt of \name]{};
	} 
	}
}
\node (S2) [left=0.25cm of 71]{$A^2:$};

\node[draw, rectangle, minimum width=20pt, minimum height=20pt, inner sep=0pt] (81)[below=0.05cm of 71]{}; 
\foreach \x in {2,...,12}
{
	\pgfmathtruncatemacro{\prev}{\x - 1}
	\pgfmathtruncatemacro{\name}{8\prev}
	
	\ifthenelse{\x=9}{
	\node[draw, rectangle, minimum width=20pt, minimum height=20pt, inner sep=0pt] (8\x) [right= 0.4cm of \name]{}; 
	}{
	\ifthenelse{\x=2 \OR \x=4 \OR \x=5 \OR \x=8 \OR \x=10 \OR \x=11}{
	\node[draw, rectangle, minimum width=20pt, minimum height=20pt, inner sep=0pt] (8\x) [right=5pt of \name]{$\checkmark$}; 
	}{
	\node[draw, rectangle, minimum width=20pt, minimum height=20pt, inner sep=0pt] (8\x) [right=5pt of \name]{};
	}
	} 
}
\node (T1) [left=0.25cm of 81]{$B^1:$};

\node[draw, rectangle, minimum width=20pt, minimum height=20pt, inner sep=0pt] (91)[below=0.05cm of 81]{}; 
\foreach \x in {2,...,12}
{
	\pgfmathtruncatemacro{\prev}{\x - 1}
	\pgfmathtruncatemacro{\name}{9\prev}
	
	\ifthenelse{\x=9}{
	\node[draw, rectangle, minimum width=20pt, minimum height=20pt, inner sep=0pt] (9\x) [right= 0.4cm of \name]{$\checkmark$}; 
	}{
	\ifthenelse{\x=2 \OR \x=3 \OR \x=6 \OR \x=8 \OR \x=9 \OR \x=12}{
	\node[draw, rectangle, minimum width=20pt, minimum height=20pt, inner sep=0pt] (9\x) [right=5pt of \name]{$\checkmark$}; 
	}{
	\node[draw, rectangle, minimum width=20pt, minimum height=20pt, inner sep=0pt] (9\x) [right=5pt of \name]{};
	} 
	}
}
\node (T2) [left=0.25cm of 91]{$B^2:$};

\end{tikzpicture}

\end{center}

\caption{An illustration of two correlated bases. Each column denotes a group of $\frac{m}{12}$ items. This construction works even if columns $1$ through $8$ denote groups of $u$ items, and columns $9$ through $12$  denote groups of $v$ items, for any $u,v$ (see \autoref{sec:sketchfinal}).}
\label{fig:bmwdouble}
\end{figure}

\subsection{Finding the Right Correlations}\label{sec:sketchfinal}

As motivated in the previous section, it is essential to have the two bases be suitably correlated to deal with the cross-terms. What is the right way to correlate these bases? It would be ideal if the cross terms coming from the `cross-pairs' $S^1, T^2$ and $S^2, T^1$ behave exactly like the terms coming from two bases $(S^1, T^1)$ and $(S^2, T^2)$. If we can make this happen, then the argument that shows why the size of the union of regular clauses from basis $1$ of Alice and regular clauses from basis $1$ of Bob is $< \nicefrac{3m}{4}$ would extend to also show that the size of the cross-terms is $< \nicefrac{3m}{4}$.

In order to show that sets $S^1, T^2$ and $S^2, T^1$ behave like bases, we need to ensure that their intersections, namely $S^1 \cap T^2$ and $S^2 \cap T^1$ have size $\nicefrac{m}{3}$, just like the intersections of two sets in a basis. Is it possible to have sets that behave in this  way?

The answer turns out to be yes, and one such construction is described in \autoref{fig:bmwdouble}. In \autoref{fig:bmwdouble}, each of the $12$ columns denotes a group of $\nicefrac{m}{12}$ items, making a total of $m$ items, and a $\checkmark$ in row $S^1$ and column $1$ means that the first $\nicefrac{m}{12}$ items are present in the set $S^1$. Importantly, note that the tuples $(S^1, T^1)$ and $(S^2, T^2)$ behave like a~\cite{BravermanMW18} basis, and have four columns in their intersection, amounting to $\nicefrac{m}{3}$ items, and so do the cross-terms $(S^1, T^2)$ and $(S^2, T^1)$.

Thus, the construction in \autoref{fig:bmwdouble} has fixed the issue with the cross-terms described in the previous section. This step is clearly necessary in order to have any hope of a successful construction, but there is one more step to ensure that knowing a $\nicefrac{3}{4}$-approximate allocation reveals which copy is special.

\paragraph{Special cross-terms.} Just like there are cross terms coming from regular clauses from basis $1$ of Alice and regular clauses from basis $2$ of Bob, there are also cross terms coming from regular clauses from basis $1$ of Alice and {\em special} clauses from basis $2$ of Bob (and vice-versa).\footnote{We do not have to deal with cross terms coming from special clauses from basis $1$ of Alice and special clauses from basis $2$ of Bob as only one of the bases will have a special clause in our construction.} 

Before we describe how we deal with these `special cross-terms', we first need to define the special clauses in our system.  We omit a precise definition in this sketch, but mention here that significant structure is imposed by the fact that special clauses need to be indistinguishable from the regular clauses. In fact, the special clauses need to more or less look like the sets $A^1$, $A^2$, $B^1$, and $B^2$ in \autoref{fig:bmwdouble}, where again a $\checkmark$ in a given column indicates that the corresponding group of $\nicefrac{m}{12}$ items is in the set. 

With this definition of special clauses, one can calculate the expected intersection of the special cross terms and check if it is $> \nicefrac{m}{4}$ or not. It turns out that with the construction in \autoref{fig:bmwdouble}, this size is exactly $\nicefrac{m}{4}$, which means that the construction does \emph{not} suffice. The reason this is problematic is because we can now simply award Alice items corresponding to an arbitrary regular clause, and Bob will get welfare $\nicefrac{m}{4}$ from its complement (using his special clause, \emph{no matter which copy his special clause is from}). 

It is here that we use the generalization of~\cite{BravermanMW18} given in \Cref{sec:bmwformal}, and let the blocks of items have unequal size. We'll assume that the first $8$ columns in \autoref{fig:bmwdouble} denote groups of $u$ items each, and the last $4$ columns denote groups of $v$ items each. For general $u$, $v$, the intersection of the regular cross terms has size:
\[\frac{5u^2v + u^3 + 6uv^2 +  2v^3}{2 (u+2v)^2(2u+v)} \cdot m.\]     

On the other hand, the intersection of a special cross terms has size:
\[\frac{16uv + 5u^2  + 6v^2}{12 (u+2v)(2u+v)} \cdot m.\]     

In fact, the parameter governing our lower bound is the minimum of the two expressions above, and this is maximized when $\nicefrac{v}{u} = 1 + \sqrt{\nicefrac{3}{2}}$. For simplicity sake, we present our main results assuming $\nicefrac{v}{u} = 2$ when the minimum of the two expressions above is $\nicefrac{61m}{240} > \nicefrac{m}{4}$. The value $\nicefrac{61m}{240}$ corresponds to the the parameter $\nicefrac{179}{240}$ in our main result.

\subsection{Summary of Outline}
So to summarize, our construction takes two correlated bases for a generalized~\cite{BravermanMW18} construction. We carefully choose the parameters of both each individual instance, as well as the correlation pattern, so that:
\begin{itemize}
\item The intersection of a regular clause of Alice and regular clause of Bob  \emph{within the same copy} is $>\nicefrac{m}{4}$.
\item The intersection of a regular clause of Alice and a regular clause of Bob \emph{across different copies} is $> \nicefrac{m}{4}$.
\item The intersection of a special clause of Alice and a regular clause of Bob \emph{from the opposite copy} is $> \nicefrac{m}{4}$.
\item It is possible to embed disjoint special clauses for both Alice and Bob within either copy, in a way so that they are indistinguishable from regular clauses.
\end{itemize}

If we can accomplish all four properties, this means that \emph{any allocation guaranteeing welfare $\geq\nicefrac{3m}{4}$ must involve at least one special clause, and a regular or special clause from the same copy}. This sketch omits the calculations, but this property suffices to guarantee that no allocation guarantees welfare $\geq \nicefrac{3m}{4}$ both when copy one is special and when copy two is special. This in turn means that knowing an allocation which guarantees welfare $\geq \nicefrac{3m}{4}$ determines which copy is special (and then careful information theoretic arguments establish that determining the special copy requires exponential communication). This completes our detailed sketch, and the technical sections confirm both that our construction satisfies the properties above, and that these properties guarantee the desired conclusion.

\section{Technical Preliminaries}\label{sec:prelim2}
This section contains notation and preliminaries necessary for our complete proofs. The following notation is standard (and some of it is previously used in our proof sketch and preliminaries), but included for completeness.

Unless otherwise specified, all logarithms are to the base $2$. We will use $\mathbb{Z}$ to denote the set of integers and $\mathbb{R}$ to denote the set of all real numbers. We also define $\mathbb{R}_+$ to denote the set of all non-negative real numbers. If $S$ is a set, then $\mathbbm{2}^S$ will denote the power set, {\em i.e.}, the set of all subsets, of $S$. Additionally, we shall denote using $S^*$ the set $\cup_{i \geq 0} S^i$, where $S^i$, for $i > 0$, is the set of all strings of length $i$ that can be formed with elements of $S$, and $S^0$ is the set containing only the empty string. The length of a string $\sigma$ will be denoted using $\len(\sigma)$.

Let $t \geq 1$ be an integer. We define $[t] = \{1,\cdots,t\}$. For a tuple $X = (X_1,\cdots,X_t)$ and integer $i \in [t]$, we define $X_{<i} = (X_1,\cdots,X_{i-1})$ and $X_{-i} = (X_1,\cdots,X_{i-1},X_{i+1},\cdots,X_t)$. 

We will use  $\unif(S)$ to denote the uniform distribution over a finite set $S$.  If $X$ is a random variable, then $\distribution{X}$ will denote the distribution of the values taken by $X$. Our proofs require careful information theoretic arguments, and Appendix~\ref{sec:info} contains thorough preliminaries for notation and facts we use.

\subsection{Partitions and Notation}\label{sec:partition}
Recall that in the~\cite{BravermanMW18} construction, one defines a distribution $D(S)$ which is uniform over all sets $A$ such that $|A \cap S| = m/3$. Such a distribution is concise to describe in text, and does not merit special notation. Our construction, however, will eventually define a distribution $\mu_\star(\cdot)$ which is uniform over all sets $A$ such that $|A \cap P_i| = p_i$ \emph{for all $i \in [16]$}. We will also frequently discuss the intersection of two sets drawn independently from such distributions, and show that it concentrates around its expectation (and compute its expectation). This section provides notation so that we can make concise descriptions and statements of this form, and concludes with a concentration inequality that we will repeatedly use. While this notation does (significantly) help keep statements concise, the reader may wish to refer back to this section for help parsing the precise statements.

We shall denote sequences with a $\vec{}\ $ on top, {\em e.g.}, $\vec{S}$. We shall use $\vec{S} \Vert \vec{S}'$ to denote the concatenation of the sequences $\vec{S}$ and $\vec{S'}$. Similarly, we shall use $\vec{S} \Vert S''$ to denote the sequence formed by appending the single element $S''$ to the sequence $\vec{S}$. Let $k > 0$ and $\vec{S} =S_1, S_2, \cdots, S_k$ be a sequence of $k$ sets. For a function $f$ defined on sets, we shall use $f(\vec{S})$ to denote the sequence $f(S_1), \cdots, f(S_k)$. Thus, $\lvert{\vec{S}}\rvert$ shall denote the sequence $\lvert{S_1}\rvert, \cdots, \lvert{S_k}\rvert$ and $\vec{S} \cap A$, for a set $A$,  shall denote the sequence $S_1 \cap A, \cdots, S_k \cap A$, {\em etc.}

Let $k > 0$. We say that a sequence $\vec{P} = P_1, P_2, \cdots, P_k$  of subsets of $M$ forms a partition of $M$ into $k$ sets if the sets $P_1, \cdots, P_k $ are pairwise disjoint and their union is $M$. Formally, it should hold that $P_i \cap P_j = \emptyset$ for all $i \neq j \in [k]$ and $\cup_{i \in [k]} P_i = M$. 
For a partition $\vec{P} = P_1, P_2, \cdots, P_k$ of $M$ into $k$ sets, and an element  $z \in M$, we define $\vec{P}[z]$ to be the unique $i \in [k]$ such that $z \in P_{i}$. Observe that our definition of a partition above ensures that  $\vec{P}[z]$ is well-defined for all $z$.

~\Cref{def:pcdist} defines the class of distributions over sets that we consider frequently throughout our construction.
\begin{definition}
\label{def:pcdist}

We say that a tuple $(k, \vec{P}, \vec{p})$ is a partition parameter if $k > 0$, $\vec{P} = P_1, \cdots, P_k$ is a partition of  $M$ into $k$ sets, and $\vec{p} = p_1, p_2, \cdots, p_k$ is a sequence of integers satisfying $0 \leq p_i \leq \lvert{P_i}\rvert$ for all $i \in [k]$.

For a partition parameter $(k, \vec{P}, \vec{p})$, we define $\pc(k, \vec{P}, \vec{p})$ to be the uniform distribution over all sets $U$ satisfying 
\[
\lvert{\vec{P} \cap U}\rvert = \vec{p}.
\]

\end{definition}


Recall in our proof sketch that we repeatedly draw regular sets from a distribution of the form $\pc(k,\vec{P},\vec{p})$, and wish to argue about the size of the intersection of two independently drawn regular sets (from different distributions). The following lemma states the expected intersection (captured in $\Delta$), and also bounds the probability that the intersection deviates far from $\Delta$. Mapping back to the~\cite{BravermanMW18} construction,~\Cref{lemma:pcdist} would help claim that all regular sets have intersection at least $7m/27-\varepsilon m$ with high probability. The proof of~\Cref{lemma:pcdist} is in~\Cref{app:partition}.

\begin{lemma}
\label{lemma:pcdist}
For any partition parameters $(k, \vec{P}, \vec{p})$ and $(k', \vec{P'}, \vec{p'})$, it holds for all $\varepsilon > 0$ that
\[
\Pr_{\substack{U \sim \pc(k, \vec{P}, \vec{p}) \\ U' \sim \pc(k', \vec{P'}, \vec{p'})}} \left(\lvert{U \cap U'}\rvert < \Delta - \varepsilon m\right) \leq  \ \exp(-\varepsilon^2 (m -  \Delta)/3 ),
\] 
where
\[
\Delta = \sum_{i \in [k]: \lvert{P_i}\rvert > 0} \sum_{i' \in [k']: \lvert{P'_{i'}}\rvert > 0} p_i p'_{i'}  \frac{\lvert{P_i \cap P'_{i'}}\rvert}{\lvert{P_i}\rvert \cdot \lvert{P'_{i'}}\rvert}.
\]

\end{lemma}

\subsubsection{The Function $\part$}
All of the partition parameters that we consider take a particular form, which enables further concise notation. Specifically, they will arise from the following construction. let $k > 0$. For any sequence $\vec{S} = S_1, \cdots, S_k$ of $k$ subsets of $M$ and any sequence $\vec{b}  = b_1, \cdots, b_k$ of bits, we define the set 
\[
\part_{\vec{S}}(\vec{b}) = \left\{z \in M \mid \forall i \in [k]: \mathbbm{1}(z \in S_i) = b_i\right\}.
\]

We use $\part_{\vec{S}}$ to denote the sequence of sets $\{\part_{\vec{S}}(\vec{b})\}_{\vec{b} \in \{0,1\}^k}$ ordered lexicographically according to $\vec{b}$ (i.e. $\part_{\vec{S}}(0^k)$, followed by $\part_{\vec{S}}(0^{k-1}1)$, etc.). 
Observe that the sequence $\part_{\vec{S}}$ forms  a partition of $M$ into $2^k$ sets.~\Cref{lemma:part} and~\Cref{cor:part} discuss marginals of distributions drawn jointly (intuitively: Alice and Bob will have inputs drawn jointly, and we will want to reason about the marginal distribution of the input that Alice sees). Applied to the~\cite{BravermanMW18} construction,~\Cref{cor:part} would be useful to claim that when $(S,T)$ are drawn uniformly at random among sets of size $m/2$ which intersect at $m/3$, that $S$ is a uniformly random set of size $m/2$. It would also be useful to claim that Alice's special set is indistinguishable from her regular sets.~\Cref{lemma:part} is a technical generalization of~\Cref{cor:part} which is necessary for our construction because we sometimes jointly draw tuples of sets (but has no analogue in~\cite{BravermanMW18}).

\begin{lemma}
\label{lemma:part}
Let $k , k_1, k_2 > 0$ and consider $\vec{a}_j \in \mathbb{Z}^{2^{k+ k_j}}$ for $j \in \{1,2\}$. Let $\vec{S}$ be a sequence of $k$ subsets of $M$. For $j \in \{1, 2\}$, define $\mu_j$ to be the uniform distribution over all sequences $\vec{S}_j$ of $k_j$ subsets of $M$ satisfying $\lvert{\part_{\vec{S} \Vert \vec{S}_j}}\rvert = \vec{a}_j$.

For  any $\vec{a} \in \mathbb{Z}^{2^{k+ k_1 + k_2}}$ such that $\Pr_{\vec{S}_1 \sim \mu_1, \vec{S}_2 \sim \mu_2}\left(\lvert{\part_{\vec{S} \Vert \vec{S}_1 \Vert \vec{S}_2}}\rvert = \vec{a}\right) > 0$, we have for all $j \in \{1, 2\}$ and all sequences $\vec{Z}$ of subsets of $M$, 
\[
\Pr_{\vec{S}_j \sim \mu_j}\left(\vec{S}_j = \vec{Z}\right) = \Pr_{\substack{\vec{S}_1 \sim \mu_1\\ \vec{S}_2 \sim \mu_2}}\left(\vec{S}_j = \vec{Z} \mid \lvert{\part_{\vec{S} \Vert \vec{S}_1 \Vert \vec{S}_2}}\rvert = \vec{a}\right).
\]

\end{lemma}

\begin{corollary}
\label{cor:part}

Let $k > 0$ and $\vec{a}_1, \vec{a}_2 \in  \mathbb{Z}^{2^{k}}$ be arbitrary. Let  $\vec{S}$ be a sequence of $k$ subsets of $M$.
 For $j \in \{1, 2\}$,  define $\mu_j:=\pc(2^k,\part_{\vec{S}},\vec{a}_j)$ (which is the uniform distribution over all sets $A \subseteq M$ satisfying $\lvert{\part_{\vec{S}} \cap A }\rvert = \vec{a}_j$).

For  any $\vec{a} \in \mathbb{Z}^{2^{k}}$ such that $\Pr_{A_1 \sim \mu_1, A_2 \sim \mu_2}\left(\lvert{\part_{\vec{S}} \cap A_1 \cap A_2 }\rvert = \vec{a}\right) > 0$, we have for all $j \in \{1, 2\}$ and all subsets $Z \subseteq M$, 
\[
\Pr_{A_j \sim \mu_j}\left(A_j = Z\right) = \Pr_{\substack{A_1 \sim \mu_1\\ A_2 \sim \mu_2}}\Big(A_j = Z \mid \lvert{\part_{\vec{S}} \cap A_1 \cap A_2 }\rvert = \vec{a} \Big).
\]

\end{corollary}

\section{Our Construction}\label{sec:dist}

For the purposes of  this section, we fix  $m > 0 $. We denote the set $[m]$ using the letter $M$.  If $S$ is a subset of $M$, then we use $\overline{S}$ to denote $M \setminus S$, {\em i.e.}, the set of items in $M$ that are {\em not} in $S$. We now give a formal definition of our lower bound instance.

\subsection{Bases and Clauses}

We next define the notion of a {\em basis}.
\begin{definition}[Basis] \label{def:basis}
A pair $S = (S^1, S^2)$ of subsets of $M$ forms a basis if 
\[
\lvert{\part_{S}}\rvert = \left(\frac{5m}{16}, \frac{3m}{16}, \frac{3m}{16}, \frac{5m}{16}\right).
\]
\end{definition}

To help parse the notation $\part_S$, recall that the first term denotes the number of elements which are in neither $S^1$ nor $S^2$ (corresponds to $\vec{b} = (0,0)$), the second term is the number of elements which are in $S^2$ but not $S^1$ (corresponds to $\vec{b} = (0,1)$), the third term is the number of elements in $S^1$ but not $S^2$ (corresponds to $\vec{b}= (1,0)$), and the fourth term is the number of elements which are in $S^1 \cap S^2$ (corresponds to $\vec{b} = (1,1)$). 

We reserve the letters $S$ and $T$ to denote bases. Note that if $S = (S^1, S^2)$ is a basis, then the pair $S^{rev} = (S^2, S^1)$ is also a basis. For notational convenience, we will treat bases as a sequence of two sets, and omit the $\vec{}\ $ sign. The following definition considers a pair of bases. Recall that $S||T$ is a list of four sets, so $|\part_{S ||T}|$ has sixteen possible $\vec{b}$ to consider (and therefore is a list of sixteen numbers).
\begin{definition}[Compatible Bases]
\label{def:compat}
We say that basis $S$ is compatible with basis $T$ if 
\[
\lvert{\part_{S \Vert T}}\rvert = \left(\frac{4m}{16}, \frac{m}{16}, 0, 0, 0, \frac{m}{16}, \frac{2m}{16}, 0, \frac{m}{16}, 0, \frac{m}{16}, \frac{m}{16}, 0, \frac{m}{16}, 0, \frac{4m}{16}\right).
\]

For shorthand, we refer by $\veccmp:=\left(\frac{4m}{16}, \frac{m}{16}, 0, 0, 0, \frac{m}{16}, \frac{2m}{16}, 0, \frac{m}{16}, 0, \frac{m}{16}, \frac{m}{16}, 0, \frac{m}{16}, 0, \frac{4m}{16}\right)$.
\end{definition}

Again, recall that (e.g.) $2m/16$ denotes the number of elements in $\bar{S^1}\cap S^2 \cap T^1 \cap \bar{T^2}$ (and corresponds to $\vec{b} = (0,1,1,0)$). An example of a basis $S$ that is compatible with  $T$ is depicted in \autoref{fig:compat}. We note that \autoref{def:compat} is not symmetric, {\em i.e.}, basis $S$ may be compatible with $T$ without basis $T$ being compatible with $S$. However, it holds that if basis $S$ is compatible with $T$, then basis $T^{rev}$ is compatible with basis $S^{rev}$.

We will use $\xi_{single}$ to denote the uniform distribution over all bases and $\xi$ to denote the uniform distribution over pairs of bases $S, T$  such that $S$ is compatible with $T$.


\textbf{The first step in our construction is the distribution $\xi$, which defines a distribution over pairs of bases.} Mapping back to our proof sketch, $(S^1,T^1)$ denotes the basis for the ``first copy,'' and $(S^2,T^2)$ denotes the basis for the ``second copy.''

\begin{figure}[h!]
\centering

\begin{tikzpicture}

\node[draw, rectangle, minimum width=12pt, minimum height=20pt, inner sep=0pt] (11){\scriptsize 1}; 
\foreach \x in {2,...,16}
{
	\pgfmathtruncatemacro{\prev}{\x - 1}
	\pgfmathtruncatemacro{\name}{1\prev}
	\ifthenelse{\x=9}{
	\node[draw, rectangle, minimum width=12pt, minimum height=20pt, inner sep=0pt] (1\x) [right= 0.4cm of \name]{\scriptsize \x}; 
	}{
	\node[draw, rectangle, minimum width=12pt, minimum height=20pt, inner sep=0pt] (1\x) [right=5pt of \name]{\scriptsize \x}; 
	}
}

\node[draw, rectangle, minimum width=12pt, minimum height=20pt, inner sep=0pt] (21)[below=0.4cm of 11]{$\checkmark$}; 
\foreach \x in {2,...,16}
{
	\pgfmathtruncatemacro{\prev}{\x - 1}
	\pgfmathtruncatemacro{\name}{2\prev}
	
	\ifthenelse{\x=9}{
	\node[draw, rectangle, minimum width=12pt, minimum height=20pt, inner sep=0pt] (2\x) [right= 0.4cm of \name]{$\checkmark$}; 
	}{
	\ifthenelse{\x=2 \OR \x=3 \OR \x=7 \OR \x=9 \OR \x=10 \OR \x=11 \OR \x=12}{
	\node[draw, rectangle, minimum width=12pt, minimum height=20pt, inner sep=0pt] (2\x) [right=5pt of \name]{$\checkmark$}; 
	}{
	\node[draw, rectangle, minimum width=12pt, minimum height=20pt, inner sep=0pt] (2\x) [right=5pt of \name]{};
	} 
	}
}
\node (S1) [left=0.25cm of 21]{$S^1:$};

\node[draw, rectangle, minimum width=12pt, minimum height=20pt, inner sep=0pt] (31)[below=0.05cm of 21]{}; 
\foreach \x in {2,...,16}
{
	\pgfmathtruncatemacro{\prev}{\x - 1}
	\pgfmathtruncatemacro{\name}{3\prev}
	
	\ifthenelse{\x=9}{
	\node[draw, rectangle, minimum width=12pt, minimum height=20pt, inner sep=0pt] (3\x) [right= 0.4cm of \name]{$\checkmark$}; 
	}{
	\ifthenelse{\x=4 \OR \x=7 \OR \x=5 \OR \x=6 \OR \x=9 \OR \x=10 \OR \x=11 \OR \x=12}{
	\node[draw, rectangle, minimum width=12pt, minimum height=20pt, inner sep=0pt] (3\x) [right=5pt of \name]{$\checkmark$}; 
	}{
	\node[draw, rectangle, minimum width=12pt, minimum height=20pt, inner sep=0pt] (3\x) [right=5pt of \name]{};
	} 
	}
}
\node (S2) [left=0.25cm of 31]{$S^2:$};

\node[draw, rectangle, minimum width=12pt, minimum height=20pt, inner sep=0pt] (41)[below=0.05cm of 31]{}; 
\foreach \x in {2,...,16}
{
	\pgfmathtruncatemacro{\prev}{\x - 1}
	\pgfmathtruncatemacro{\name}{4\prev}
	
	\ifthenelse{\x=9}{
	\node[draw, rectangle, minimum width=12pt, minimum height=20pt, inner sep=0pt] (4\x) [right= 0.4cm of \name]{$\checkmark$}; 
	}{
	\ifthenelse{\x=2 \OR \x=3 \OR \x=4 \OR \x=5 \OR \x=9 \OR \x=10 \OR \x=11 \OR \x=12}{
	\node[draw, rectangle, minimum width=12pt, minimum height=20pt, inner sep=0pt] (4\x) [right=5pt of \name]{$\checkmark$}; 
	}{
	\node[draw, rectangle, minimum width=12pt, minimum height=20pt, inner sep=0pt] (4\x) [right=5pt of \name]{};
	}
	} 
}
\node (T1) [left=0.25cm of 41]{$T^1:$};

\node[draw, rectangle, minimum width=12pt, minimum height=20pt, inner sep=0pt] (51)[below=0.05cm of 41]{}; 
\foreach \x in {2,...,16}
{
	\pgfmathtruncatemacro{\prev}{\x - 1}
	\pgfmathtruncatemacro{\name}{5\prev}
	
	\ifthenelse{\x=9}{
	\node[draw, rectangle, minimum width=12pt, minimum height=20pt, inner sep=0pt] (5\x) [right= 0.4cm of \name]{$\checkmark$}; 
	}{
	\ifthenelse{\x=2 \OR \x=7 \OR \x=6 \OR \x=8 \OR \x=9 \OR \x=10 \OR \x=11 \OR \x=12}{
	\node[draw, rectangle, minimum width=12pt, minimum height=20pt, inner sep=0pt] (5\x) [right=5pt of \name]{$\checkmark$}; 
	}{
	\node[draw, rectangle, minimum width=12pt, minimum height=20pt, inner sep=0pt] (5\x) [right=5pt of \name]{};
	} 
	}
}
\node (T2) [left=0.25cm of 51]{$T^2:$};

\node[draw, rectangle, minimum width=12pt, minimum height=20pt, inner sep=0pt] (61)[below=0.4cm of 51]{$\checkmark$}; 
\foreach \x in {2,...,16}
{
	\pgfmathtruncatemacro{\prev}{\x - 1}
	\pgfmathtruncatemacro{\name}{6\prev}
	
	\ifthenelse{\x=9}{
	\node[draw, rectangle, minimum width=12pt, minimum height=20pt, inner sep=0pt] (6\x) [right= 0.4cm of \name]{$\checkmark$}; 
	}{
	\ifthenelse{\x=3 \OR \x=6 \OR \x=7 \OR \x=9 \OR \x=10 \OR \x=15 \OR \x=16}{
	\node[draw, rectangle, minimum width=12pt, minimum height=20pt, inner sep=0pt] (6\x) [right=5pt of \name]{$\checkmark$}; 
	}{
	\node[draw, rectangle, minimum width=12pt, minimum height=20pt, inner sep=0pt] (6\x) [right=5pt of \name]{};
	} 
	}
}
\node (S1) [left=0.25cm of 61]{$A^1_{\star}:$};

\node[draw, rectangle, minimum width=12pt, minimum height=20pt, inner sep=0pt] (71)[below=0.05cm of 61]{$\checkmark$}; 
\foreach \x in {2,...,16}
{
	\pgfmathtruncatemacro{\prev}{\x - 1}
	\pgfmathtruncatemacro{\name}{7\prev}
	
	\ifthenelse{\x=9}{
	\node[draw, rectangle, minimum width=12pt, minimum height=20pt, inner sep=0pt] (7\x) [right= 0.4cm of \name]{}; 
	}{
	\ifthenelse{\x=4 \OR \x=5 \OR \x=7 \OR \x=11 \OR \x=12 \OR \x=13 \OR \x=14}{
	\node[draw, rectangle, minimum width=12pt, minimum height=20pt, inner sep=0pt] (7\x) [right=5pt of \name]{$\checkmark$}; 
	}{
	\node[draw, rectangle, minimum width=12pt, minimum height=20pt, inner sep=0pt] (7\x) [right=5pt of \name]{};
	} 
	}
}
\node (T1) [left=0.25cm of 71]{$A^2_{\star}:$};

\node[draw, rectangle, minimum width=12pt, minimum height=20pt, inner sep=0pt] (81)[below=0.05cm of 71]{}; 
\foreach \x in {2,...,16}
{
	\pgfmathtruncatemacro{\prev}{\x - 1}
	\pgfmathtruncatemacro{\name}{8\prev}
	
	\ifthenelse{\x=9}{
	\node[draw, rectangle, minimum width=12pt, minimum height=20pt, inner sep=0pt] (8\x) [right= 0.4cm of \name]{}; 
	}{
	\ifthenelse{\x=2 \OR \x=4 \OR \x=5 \OR \x=8 \OR \x=11 \OR \x=12 \OR \x=13 \OR \x=14}{
	\node[draw, rectangle, minimum width=12pt, minimum height=20pt, inner sep=0pt] (8\x) [right=5pt of \name]{$\checkmark$}; 
	}{
	\node[draw, rectangle, minimum width=12pt, minimum height=20pt, inner sep=0pt] (8\x) [right=5pt of \name]{};
	}
	} 
}
\node (S2) [left=0.25cm of 81]{$B^1_{\star}:$};

\node[draw, rectangle, minimum width=12pt, minimum height=20pt, inner sep=0pt] (91)[below=0.05cm of 81]{}; 
\foreach \x in {2,...,16}
{
	\pgfmathtruncatemacro{\prev}{\x - 1}
	\pgfmathtruncatemacro{\name}{9\prev}
	
	\ifthenelse{\x=9}{
	\node[draw, rectangle, minimum width=12pt, minimum height=20pt, inner sep=0pt] (9\x) [right= 0.4cm of \name]{$\checkmark$}; 
	}{
	\ifthenelse{\x=2 \OR \x=3 \OR \x=6 \OR \x=8 \OR \x=9 \OR \x=10 \OR \x=15 \OR \x=16}{
	\node[draw, rectangle, minimum width=12pt, minimum height=20pt, inner sep=0pt] (9\x) [right=5pt of \name]{$\checkmark$}; 
	}{
	\node[draw, rectangle, minimum width=12pt, minimum height=20pt, inner sep=0pt] (9\x) [right=5pt of \name]{};
	} 
	}
}
\node (T2) [left=0.25cm of 91]{$B^2_{\star}:$};

\end{tikzpicture}

\caption{A basis $S = (S^1, S^2)$ that is compatible with another basis  $T = (T^1, T^2)$. Also pictured: a pair of sets $(A^1_{\star},A^2_{\star})$ special with respect to $(S,T)$ (see \Cref{sec:special}). Observe that $(B^2_{\star},B^1_{\star}) = (\overline{A^2_{\star}},\overline{A^1_{\star}})$ is special with respect to $(T^{rev}, S^{rev})$. Here, the blocks inside each column correspond to the same $m/16$ elements.}
\label{fig:compat}
\end{figure}

\subsubsection{Regular Clauses} \label{sec:regular}
The next step in our construction is to define how to draw regular clauses, once the bases are fixed. In order to have the desired interaction between cross terms, we need to specify the intersection of each clause not only with the basis ``of its copy'', but also the basis for the ``other copy.''

\begin{definition}[Clause] 
\label{def:clause} Let $ S= (S^1, S^2)$ be basis. We say that a set $A \subseteq M$ is a clause with respect to $S$ if 
\[
\lvert{\part_S \cap A}\rvert = \left(\frac{2m}{16}, \frac{m}{16}, \frac{2m}{16}, \frac{3m}{16}\right).
\]
For shorthand, we denote by $\vecreg:=\left(\frac{2m}{16}, \frac{m}{16}, \frac{2m}{16}, \frac{3m}{16}\right)$.
\end{definition}
We define $\mu_{single}(S)$ to be the uniform distribution over all clauses with respect to $S$. Observe that the distribution $\mu_{single}(S) = \pc(4,\part_S, \vecreg  )$ (recall the definition of $\pc$ from Definition~\ref{def:pcdist}). We also define:
\begin{definition}[The distribution $\mu(\cdot)$]
\label{def:mu} Let $ S= (S^1, S^2)$ be a basis. A pair $(A^1, A^2)$ of subsets of $M$ is called a clause pair with respect to $S$ if $A^1$ is a clause with respect to $S$, $A^2$ is a clause with respect to $S^{rev}$ and we have 
\[
\lvert{\part_S \cap A^1 \cap A^2}\rvert = \left(0, 0, \frac{m}{16}, \frac{m}{16}\right).
\]

For shorthand, we define $\vecregpair:=\left(0, 0, \frac{m}{16}, \frac{m}{16}\right)$.

We define $\mu(S)$ to be the uniform distribution over all clause pairs with respect to $S$.
\end{definition}

\textbf{The second step in our construction is the distribution $\mu(\cdot)$, which describes how Alice and Bob draw pairs of regular clauses once their basis is fixed.}

Observe that $\mu(S)$ is a distribution over \emph{pairs} of clauses. The first clause in the pair is a clause with respect to $S$ (this corresponds to a regular clause in the ``first copy''), and the second is a clause with respect to $S^{rev}$ (this corresponds to a regular clause in the ``second copy''). Observation~\ref{obs:mu} below is simple, but key: it states that a pair of sets $(A^1,A^2)$ is a clause pair with respect to $S$ if and only if a sequence of equalities involving the size of sets involving $S, A^1, A^2$ holds. Because $\mu(S)$ is the uniform distribution clause pairs with respect to $S$, this means that any $(A^1,A^2)$ satisfying the noted equalities is equally likely to have been drawn from $\mu(S)$ (and this is what lets us later plant an undetectable special clause pair). 

\begin{observation}
\label{obs:mu}
Observe that for any basis $S$, the fact that a pair of sets $(A^1, A^2)$ is a clause pair with respect to $S$ implies that  $\lvert{\part_{S}}\rvert $, $\lvert{\part_{S} \cap A^1}\rvert $, $\lvert{\part_{S} \cap A^2}\rvert $, and $\lvert{\part_{S } \cap A^1 \cap A^2}\rvert $ are all fixed functions of $m$. This means that there exist a vector $\vecpair$ such that $(A^1, A^2)$ is a clause pair with respect to $S$ if and only if 
\[
\lvert{\part_{S \Vert A^1 \Vert A^2}}\rvert  = \vecpair.
\]
\end{observation}

In our lower bound construction, Alice's regular clauses are drawn from the distribution $\mu(S)$ while Bob's regular clauses are drawn from the distribution $\mu(T)$, where $S$ and $T$ are bases such that $S$ is compatible with $T$. The following lemma shows that the intersection of a regular clause of Alice and a regular clause of Bob has size at least $\frac{51m}{200} > \frac{m}{4}$ (with high probability). While the proof requires several steps to be rigorous, the intuition is simple: we first need to argue that each of the sets $A^1,A^2,B^1,B^2$ are identically distributed to draws from a distribution of the form $\mu_{single}(\cdot)$, which is of the form $\pc(k,\vec{P},\vec{p})$. This step uses~\Cref{lemma:part}. Once we have done this, we can use~\Cref{lemma:pcdist} to argue that the intersection of any two pairs concentrates around its expectation (and that its expectation is $51m/100$).

\begin{lemma}
\label{lemma:reg} Consider $\varepsilon > 0$ and bases $S, T$  such that $S$ is compatible with $T$. For all $i,j \in \{1, 2\}$, we have
\[
\Pr_{\substack{(A^1, A^2) \sim \mu(S)\\ (B^2, B^1) \sim \mu(T^{rev})}}\left( \lvert{A^i \cap B^j}\rvert < \frac{51m}{200} - \varepsilon m\right) \leq \exp(-\varepsilon^2 m/20).
\]
\end{lemma}
\begin{proof} We show the lemma assuming $i = j = 1$. The proof for other values of $(i,j)$ is similar (with different calculations), and we discuss necessary modifications at the end. We derive:
\begin{align*}
\Pr_{\substack{(A^1, A^2) \sim \mu(S)\\ (B^2, B^1) \sim \mu(T^{rev})}}&\left( \lvert{A^1 \cap B^1}\rvert < \frac{51m}{200} - \varepsilon m\right)  \\
&= \sum_{Z, Z' : \lvert{Z \cap Z'}\rvert < \frac{51m}{200} - \varepsilon m}\Pr_{\substack{(A^1, A^2) \sim \mu(S)\\ (B^2, B^1) \sim \mu(T^{rev})}}\left((A^1, B^1) = (Z, Z')\right)\\
&= \sum_{Z, Z' : \lvert{Z \cap Z'}\rvert < \frac{51m}{200} - \varepsilon m}\Pr_{(A^1, A^2) \sim \mu(S)}(A^1 = Z)\Pr_{(B^2, B^1) \sim \mu(T^{rev})}(B^1 = Z') \\
&= \sum_{Z, Z' : \lvert{Z \cap Z'}\rvert < \frac{51m}{200} - \varepsilon m}\Pr_{ \substack{A^1 \sim \mu_{single}(S)\\ A^2 \sim \mu_{single}(S^{rev}) }}(A^1 = Z \mid\lvert{\part_S \cap A^1 \cap A^2}\rvert = \vecregpair) \\
&\hspace{4cm} \times\Pr_{ \substack{B^2 \sim \mu_{single}(T^{rev})\\ B^1 \sim \mu_{single}(T) }}(B^1 = Z' \mid \lvert{\part_{T^{rev}} \cap B^2 \cap B^1}\rvert = \vecregpair) \\
&= \sum_{Z, Z' : \lvert{Z \cap Z'}\rvert < \frac{51m}{200} - \varepsilon m}\Pr_{A \sim \mu_{single}(S)}(A = Z)\Pr_{B \sim \mu_{single}(T)}(B = Z') \tag{\autoref{cor:part}} \\
&= \sum_{Z, Z' : \lvert{Z \cap Z'}\rvert < \frac{51m}{200} - \varepsilon m}\Pr_{\substack{A \sim \mu_{single}(S)\\ B \sim \mu_{single}(T)}}\left((A, B) = (Z, Z')\right) \\
&= \Pr_{\substack{A \sim \mu_{single}(S)\\ B \sim \mu_{single}(T)}}\left(\lvert{A \cap B}\rvert < \frac{51m}{200} - \varepsilon m\right),
\end{align*}
It is thus sufficient to show that $\Pr_{A \sim \mu_{single}(S), B \sim \mu_{single}(T)}\left(\lvert{A \cap B}\rvert < \frac{51m}{200} - \varepsilon m\right) \leq \exp(-\varepsilon^2 m/20)$. We show this using \autoref{lemma:pcdist} as the distributions $\mu_{single}(S) = \pc\left(4, \part_{S}, \vecreg\right)$ and $\mu_{single}(T) = \pc\left(4, \part_{T}, \vecreg\right)$.  By \autoref{lemma:pcdist}, we have 
\[
\Pr_{\substack{A \sim \mu_{single}(S)\\ B \sim \mu_{single}(T)}}\left(\lvert{A \cap B}\rvert < \Delta - \varepsilon m\right) \leq \exp(-\varepsilon^2 (m - \Delta)/3),
\]
 so we just need to compute $\Delta$. Below, recall that $\veccmp$ lists the size of $S^1\cap S^2\cap T^1\cap T^2$, $S^1 \cap S^2 \cap T^1 \cap \bar{T^2}$, etc., and this is where the terms $\frac{4m}{16},\frac{m}{16}$, etc. come from. Recall that $\vecreg$ lists the size of $A^1 \cap S^1 \cap S^2$, etc., and also $B^1 \cap T^1 \cap T^2$, etc. So for example, $|A^1 \cap \bar{S^1} \cap \bar{S^2}| = 2m/16$ (according to $\vecreg$, and $|\bar{S^1} \cap \bar{S^2}| = 4m/16+m/16+0+0 = 5m/16$ (according to $\veccmp$), and therefore $A^1$ contains $2/5$ of the elements in $\bar{S^1} \cap \bar{S^2}$. 

\begin{align*}
\Delta &=  \frac{2}{5} \cdot \frac{2}{5} \cdot \frac{4m}{16} + \frac{2}{5} \cdot \frac{1}{3} \cdot \frac{m}{16} + \frac{1}
{3} \cdot \frac{1}{3} \cdot \frac{m}{16}  \\
&\hspace{1cm}+ \frac{1}{3} \cdot \frac{2}{3} \cdot \frac{2m}{16} + \frac{2}{3} \cdot \frac{2}{5} \cdot \frac{m}{16} +\frac{2}{3} \cdot \frac{2}{3} \cdot \frac{m}{16}  \\
&\hspace{1cm}+ \frac{2}{3} \cdot \frac{3}{5} \cdot \frac{m}{16} + \frac{3}{5} \cdot \frac{1}{3} \cdot \frac{m}{16} +\frac{3}{5} \cdot \frac{3}{5} \cdot \frac{4m}{16}  \\
&= \frac{51}{200} \cdot m.
\end{align*}
Thus, we get,
\[
\Pr_{\substack{A \sim \mu_{single}(S)\\ B \sim \mu_{single}(T)}}\left(\lvert{A \cap B}\rvert < \frac{51m}{200} - \varepsilon m\right) \leq \exp(-149 \varepsilon^2 m/600) <   \exp(-\varepsilon^2 m/20) ,
\]
as desired.

To adjust the proof for the other three values of $(i,j)$, the first half of the proof would be identical, but perhaps replacing $\mu_{single}(S)$ with $\mu_{single}(S^{rev})$ and perhaps replacing $\mu_{single}(T)$ with $\mu_{single}(T^{rev})$. This also causes the precise calculations above for $\Delta$ to change, but all four calculations result in $\Delta \geq 51m/200$.
\end{proof}

\textbf{\Cref{lemma:reg} is the first key property of our construction, which establishes that the union of two regular clauses is $<3m/4$}. Note in particular that Lemma~\ref{lemma:reg} covers both the ``like terms'' and the ``cross terms'' at once. Note also that if we were have a construction which draws uniformly random compatible bases from $\xi$, and then has Alice and Bob draw exponentially-many (but not too many) clause pairs with respect to their basis, that the optimal welfare would be at most $149m/200$.

\subsubsection{Special Clauses} \label{sec:special}
We now describe how to add special clauses to our construction. Again recall that there are three properties we need: first, the special clauses should be indistinguishable from regular clauses. Second, Alice and Bob's special clauses should be disjoint. Third, a special clause should intersect a regular clause ``from the other copy'' at slightly more than $m/4$.

\begin{definition} [Special clauses]
\label{def:special}Let $ S, T$ be bases such that $S$ is compatible with $T$.  We say that a set $A_{\star} \subseteq M$ is $1$-special with respect to $(S, T)$ if:
\[
\lvert{\part_{S \Vert T} \cap A_{\star}}\rvert = \left(\frac{2m}{16}, 0, 0, 0, 0, \frac{m}{16}, 0, 0, \frac{m}{16}, 0, \frac{m}{16}, 0, 0, \frac{m}{16}, 0, \frac{2m}{16}\right).
\]
Similarly, we say that $A_{\star}$ is $2$-special with respect to $(S, T)$ if:
\[
\lvert{\part_{S \Vert T} \cap A_{\star}}\rvert = \left(\frac{2m}{16}, 0, 0, 0, 0, 0, \frac{2m}{16}, 0, \frac{m}{16}, 0, 0, 0, 0, \frac{m}{16}, 0, \frac{2m}{16}\right).
\]

For shorthand, we refer by $\vecspec_1:=\left(\frac{2m}{16}, 0, 0, 0, 0, \frac{m}{16}, 0, 0, \frac{m}{16}, 0, \frac{m}{16}, 0, 0, \frac{m}{16}, 0, \frac{2m}{16}\right)$, and $\vecspec_2:= \left(\frac{2m}{16}, 0, 0, 0, 0, 0, \frac{2m}{16}, 0, \frac{m}{16}, 0, 0, 0, 0, \frac{m}{16}, 0, \frac{2m}{16}\right)$.
\end{definition}

For $i \in \{1, 2\}$, we define $\mu^i_{\star, single}(S, T)$ to be the uniform distribution over all sets that are $i$-special with respect to $(S, T)$.  Observe that $\mu^i_{\star, single}(S, T) = \pc\left(16, \part_{S \Vert T}, \vecspec_i \right)$ for  $i \in \{1, 2\}$. We again define a distribution over a pair of special sets (again intuitively, $A^1_\star$ is special for the ``first copy'' and $A^2_\star$ is special for the ``second copy'').

\begin{definition}[The distribution $\mu_{\star}(\cdot)$]
\label{def:mustar} Let $ S, T$ be bases such that $S$ is compatible with $T$.  
We say that a pair of sets $(A^1_{\star}, A^2_{\star})$ is special with respect to $(S, T)$ if $A^1_{\star}$ is $1$-special with respect to $(S, T)$ and $A^2_{\star}$ is $2$-special with respect to $(S, T)$ and 
\[
\lvert{\part_{S \Vert T} \cap A^1_{\star} \cap A^2_{\star}}\rvert = \left(0, 0, 0, 0, 0, 0, 0, 0, \frac{m}{16}, 0, 0, 0, 0, \frac{m}{16}, 0, 0\right).
\]

For shorthand, we define $\vecspecpair:=\left(0, 0, 0, 0, 0, 0, 0, 0, \frac{m}{16}, 0, 0, 0, 0, \frac{m}{16}, 0, 0\right)$. We define $\mu_{\star}(S, T)$ to be the uniform distribution over all pairs of sets that are special with respect to $(S, T)$.

\end{definition}

\textbf{The third step in our construction is the distribution $\mu_\star(\cdot)$, which describes how Alice and Bob draw potential special clauses once their basis is fixed.}~\Cref{obs:mustar} is again simple, but crucial. In particular, it observes that every pair that is special with respect to $(S,T)$ is also a clause pair with respect to $S$. This means that an independently drawn special pair will be indistinguishable from clause pairs.

\begin{observation}
\label{obs:mustar}
Observe that for bases $S, T$ such that $S$ is compatible with $T$, the fact that a pair of sets $(A^1_{\star}, A^2_{\star})$ is special with respect to $(S, T)$ implies that  $\lvert{\part_{S \Vert T}}\rvert $, $\lvert{\part_{S \Vert T} \cap A^1_{\star}}\rvert $, $\lvert{\part_{S \Vert T} \cap A^2_{\star}}\rvert $, and $\lvert{\part_{S \Vert T} \cap A^1_{\star} \cap A^2_{\star}}\rvert $ are all fixed functions of $m$. This means that there exist a vector $\vecopt$ such that $(A^1_{\star}, A^2_{\star})$ is special with respect to $(S, T)$ if and only if 
\[
\lvert{\part_{S \Vert T \Vert A^1_{\star} \Vert A^2_{\star}}}\rvert  = \vecopt.
\]
We reserve $\vecopt$ to denote this vector for the rest of this document. Furthermore, observe that any pair $(A^1_{\star}, A^2_{\star})$ that  is special with respect to $(S, T)$ is a clause pair with respect to $S$. Thus, for all $Z^1, Z^2 \subseteq M$, we have that 
\[
\Pr_{(A^1_{\star}, A^2_{\star}) \sim \mu_{\star}(S, T)}\left((A^1_{\star}, A^2_{\star}) = (Z^1, Z^2)\right) = \Pr_{(A^1, A^2) \sim \mu(S)}\left((A^1, A^2) = (Z^1, Z^2) \mid \lvert{\part_{S \Vert T \Vert A^1 \Vert A^2}}\rvert = \vecopt\right) .
\]
\end{observation}

\textbf{\Cref{obs:mustar} is the second key property of our construction, which suggests that special clauses are indistinguishable from regular clauses, prior to any communication.} Recall that if $S$ is compatible with $T$, then $T^{rev}$ is compatible with $S^{rev}$. It can be verified from \autoref{def:mustar} that $(A^1_{\star},A^2_{\star})$ is special with respect to $(S,T)$ if and only if $(\overline{A^2_{\star}},\overline{A^1_{\star}})$ is special with respect to $(T^{rev},S^{rev})$.  See \autoref{fig:compat} for a depiction of such a configuration of sets.

Next, we show, in \autoref{lemma:special}, an analogue of  \autoref{lemma:reg} for special sets. Just like \autoref{lemma:reg} shows that the intersection of a regular clause of Alice and a regular clause of Bob has size  $> \frac{m}{4}$ with high probability, \autoref{lemma:special} shows that if $(A^1_{\star},A^2_{\star})$ is special with respect to $(S,T)$, then, intersection of $A^1_{\star}$ with any clause with respect to $T^{rev}$ and the intersection of $A^2_{\star}$ with any clause with respect to $T$ has size $ > \frac{m}{4}$ with high probability. 

\emph{We note that \autoref{lemma:special} does not make similar claims regarding the intersection of $A^1_{\star}$ and clauses with respect to $T$ and the intersection of $A^2_{\star}$ and clauses with respect to $T^{rev}$}. This is no coincidence, as these intersections have size $< \frac{m}{4}$ (with high probability). Intuitively, this should be expected: recall from the~\cite{BravermanMW18} construction that a special clause for Alice and a regular clause for Bob had intersection $< m/4$. The intersection of $A^1_{\star}$ with a clause with respect to $T$ is the analogue in our construction. But we still need to make sure that the intersection of a special clause of Alice for one copy and a regular clause for Bob \emph{in the other copy} is large, and this is what~\Cref{lemma:special} states.

\begin{lemma}
\label{lemma:special} Consider $\varepsilon > 0$ and bases $S, T$  such that $S$ is compatible with $T$. For all $i \in \{1, 2\}$, we have
\[
\Pr_{\substack{(A^1_{\star},A^2_{\star}) \sim \mu_{\star}(S, T)\\ (B^2, B^1) \sim \mu(T^{rev})}}\left( \lvert{A^i_{\star} \cap B^{3-i}}\rvert < \frac{61m}{240} - \varepsilon m\right) \leq \exp(-\varepsilon^2 m/20).
\]
\end{lemma}
\begin{proof} We show the lemma assuming $i = 1$. The proof for $i=2$ is similar (with different calculations), and we discuss necessary modifications at the end. We derive:
\begin{align*}
\Pr_{\substack{(A^1_{\star},A^2_{\star}) \sim \mu_{\star}(S, T)\\ (B^2, B^1) \sim \mu(T^{rev})}}&\left( \lvert{A^1_{\star} \cap B^2}\rvert < \frac{61m}{240} - \varepsilon m\right) \\
&= \sum_{Z, Z' : \lvert{Z \cap Z'}\rvert < \frac{61m}{240} - \varepsilon m}\Pr_{\substack{(A^1_{\star},A^2_{\star}) \sim \mu_{\star}(S, T)\\ (B^2, B^1) \sim \mu(T^{rev})}} \left((A^1_{\star}, B^2) = (Z, Z')\right)\\
&= \sum_{Z, Z' : \lvert{Z \cap Z'}\rvert < \frac{61m}{240} - \varepsilon m}\Pr_{(A^1_{\star},A^2_{\star}) \sim \mu_{\star}(S, T)} \left(A^1_{\star} = Z\right)\Pr_{(B^2, B^1) \sim \mu(T^{rev})} \left( B^2 = Z'\right)\\
&= \sum_{Z, Z' : \lvert{Z \cap Z'}\rvert < \frac{61m}{240} - \varepsilon m}\Pr_{ \substack{A^1_{\star} \sim \mu^1_{\star, single}(S, T) \\ A^2_{\star} \sim \mu^2_{\star, single}(S, T)}} \left(A^1_{\star} = Z \mid \lvert{\part_{S \Vert T} \cap A^1_{\star} \cap A^2_{\star}}\rvert = \vecspecpair \right) \\
&\hspace{4cm} \times \Pr_{\substack{B^2 \sim \mu_{single}(T^{rev}) \\ B^1 \sim \mu_{single}(T)}} \left( B^2 = Z'  \mid \lvert{\part_{T^{rev}} \cap B^2 \cap B^1}\rvert = \vecregpair \right)\\
&= \sum_{Z, Z' : \lvert{Z \cap Z'}\rvert < \frac{61m}{240} - \varepsilon m}\Pr_{A_{\star} \sim \mu^1_{\star, single}(S, T)} \left(A_{\star} = Z\right)\Pr_{B \sim \mu_{single}(T^{rev})} \left( B = Z'\right)\tag{\autoref{cor:part}} \\
&= \sum_{Z, Z' : \lvert{Z \cap Z'}\rvert < \frac{61m}{240} - \varepsilon m}\Pr_{\substack{A_{\star} \sim \mu^1_{\star, single}(S, T)\\ B \sim \mu_{single}(T^{rev}) }} \left((A_{\star}, B) = (Z,Z')\right)\\
&=\Pr_{\substack{A_{\star} \sim \mu^1_{\star, single}(S, T)\\ B \sim \mu_{single}(T^{rev}) }} \left(\lvert{A_{\star} \cap B}\rvert < \frac{61m}{240} - \varepsilon m\right).
\end{align*}
It is thus sufficient to show that $\Pr_{A_{\star} \sim \mu^1_{\star, single}(S, T), B \sim \mu_{single}(T^{rev}) }\left(\lvert{A_{\star} \cap B}\rvert < \frac{61m}{240} - \varepsilon m\right) \leq \exp(-\varepsilon^2 m/20)$. We show this using \autoref{lemma:pcdist} as the distribution $ \mu^1_{\star, single}(S, T)  = \pc\left(16, \part_{S \Vert T}, \vecspec_1 \right)$ and $\mu_{single}(T^{rev}) = \pc\left(4, \part_{T^{rev}}, \vecreg\right)$.  By \autoref{lemma:pcdist}, we have 
\[
\Pr_{\substack{A_{\star} \sim \mu^1_{\star, single}(S, T)\\ B \sim \mu_{single}(T^{rev}) }} \left(\lvert{A_{\star} \cap B}\rvert < \Delta - \varepsilon m\right)  \leq \exp(-\varepsilon^2 (m - \Delta)/3),
\]
so we just need to compute $\Delta$. Again, recall that the relevant terms come from the vectors $\vecspec_1,\vecreg,\veccmp$. Expanding the calculations, we get:
\begin{align*}
\Delta &=  \frac{1}{2} \cdot \frac{2}{5} \cdot \frac{4m}{16} + \frac{2}{3} \cdot \frac{m}{16} + \frac{2}{5} \cdot \frac{m}{16}  \\
&\hspace{1cm}+ \frac{1}{3} \cdot \frac{m}{16} +  \frac{2}{3} \cdot \frac{m}{16}  +\frac{1}{2} \cdot \frac{3}{5} \cdot \frac{4m}{16}  = \frac{61}{240} \cdot m.
\end{align*}
Thus, we get,
\[
\Pr_{\substack{A_{\star} \sim \mu^1_{\star, single}(S, T)\\ B \sim \mu_{single}(T^{rev}) }} \left(\lvert{A_{\star} \cap B}\rvert < \frac{61m}{240} - \varepsilon m\right)  \leq \exp(-179 \varepsilon^2 m/720) <   \exp(-\varepsilon^2 m/20) ,
\]
as desired. 

Adjusting the proof for $i=2$ just requires replacing $1$ with $2$ in the first half of the proof. The calculations for $\Delta$ are similar, and also $\geq 61m/240$.
\end{proof}

\textbf{\Cref{lemma:special} is the third key property of our construction, which establishes that the union of a special clause for one copy with a regular clause of the other copy is $< 3m/4$.} 

With the three building blocks and these three properties, we can now define our full construction.

\subsection{The Distribution $\nu$} \label{sec:nu}

We now define a distribution $\nu$ over pairs of functions $(v^{\alice} , v^{\bob}) \in \bxos_m$ (recall the definition of $\bxos_m$ from \Cref{sec:mainformal}) that we will use to show \autoref{thm:mainred}. 

Fix $\varepsilon > 0$ and define $n = \exp\left(\frac{\eps^2 \cdot m}{100}\right)$. We assume for simplicity that $n$ is an integer. This will be our hard instance for $\bxos_m$ combinatorial auctions.

\begin{tbox}
\begin{itemize}[leftmargin=10pt]
	\item \textbf{Sampling $(v^{\alice} , v^{\bob}) \sim \nu$}: 
\begin{enumerate}[label=(\arabic*)]
	\item \label{line:nubases} Sample bases $(S, T) \sim \xi$.
	\item \label{line:nuclauses} Sample $i_{\star} \sim \unif([n])$ and construct sequences $\vec{A}^1, \vec{A}^2, \vec{B}^1, \vec{B}^2$ of $n$ subsets of $M$ as follows (where $\vec{A}^1 = A^1_1, \cdots, A^1_n$, {\em etc.}):
	\begin{enumerate}[label=(\alph*)]
		\item For $i \neq i_{\star} \in [n]$,  sample $(A^1_i, A^2_i)  \sim \mu(S)$ and $(B^2_i, B^1_i)  \sim \mu(T^{rev})$ independently. 
		\item Sample $(A^1_{\star}, A^2_{\star}) \sim \mu_{\star}(S, T)$ and set $(A^1_{i_{\star}}, A^2_{i_{\star}}, B^1_{i_{\star}}, B^2_{i_{\star}})  =  (A^1_{\star}, A^2_{\star}, \overline{A^1_{\star}}, \overline{A^2_{\star}})$. 
	\end{enumerate}
	\item Sample $\theta \in \unif(\{1,2\})$, and sequences $\vec{r}^{\alice}=r^{\alice}_1,\cdots,r^{\alice}_n \in \{1,2\}^{n}$ and $\vec{r}^{\bob}=r^{\bob}_1,\cdots,r^{\bob}_n \in \{1,2\}^{n}$ uniformly at random subject to $r^{\alice}_{i_{\star}} = r^{\bob}_{i_{\star}} = \theta$.
	\item Define $v^{\alice}(Z) = \max_{F \in \mathcal{F}^{\alice}} \vert{Z \cap F}\rvert$  and $v^{\bob}(Z) = \max_{F \in \mathcal{F}^{\bob}} \vert{Z \cap F}\rvert$ where, for all $Z \subseteq M$,
	\[
		\mathcal{F}^{\alice} = \{A^{r^{\alice}_i}_i \mid i \in [n]\} \hspace{1cm} \text{ and }\hspace{1cm} \mathcal{F}^{\bob} = \{B^{r^{\bob}_i}_i \mid i \in [n]\}.
	\]
\end{enumerate}
\end{itemize}
\end{tbox}

Before continuing, we briefly elaborate on each step, and connect it to our proof sketch. In (1), we jointly draw a basis for each copy of the modified~\cite{BravermanMW18} construction. $(S^1,T^1)$ is the basis for the first copy, and $(S^2,T^2)$ is the basis for the second copy. In step (2), we first draw a uniformly random index in $[n]$ where we will hide the special clauses. Each index $i$ corresponds to \emph{two} clauses for Alice and \emph{two} clauses for Bob. Intuitively, the first clause for Alice is in ``copy one'' and the second is in ``copy two.'' In (2a), we draw pairs of regular clauses uniformly at random for each non-special index for both Alice and Bob. In (2b) we jointly draw special clauses for Alice and Bob \emph{that are disjoint}. In step (3), we visit each index and pick \emph{one of the two clauses uniformly at random to include}. That is, for each index, there is a ``copy one'' clause and a ``copy two'' clause. One of these will be a clause in the defined valuation function in step (4), and one of them will be ignored. Importantly, \emph{$r^\alice_{i^\star}=r^\bob_{i^\star}=\theta$}, meaning that Alice and Bob have a special set from the same copy, and therefore the optimal welfare is $m$ in every instance drawn from $\nu$. This further implies that knowing $\theta$ is equivalent to know which copy is special. This setup allows us to provide a somewhat clean outline of an information theoretic proof that learning $\theta$ requires exponential communication---$r^\alice_{i^\star}$ appears indistinguishable from $r^\alice_{i}$ for all other $i \in [n]$. Therefore, any \emph{simultaneous} algorithm which reveals non-trivial information about $r^\alice_{i^\star}$ must reveal non-trivial information about all $r^\alice_i$. We now proceed with analysis of our construction.

For notational convenience, it will be easier to consider $\nu$ as the  distribution of a random variable $\Upsilon  = (S, T, i_{\star}, \vec{A}^1, \vec{A}^2, \vec{B}^1, \vec{B}^2, \theta, \vec{r}^{\alice}, \vec{r}^{\bob})$ and consider $v^{\alice}, v^{\bob}$ as functions of $\Upsilon$. We will also need shorthand for certain entries of $\Upsilon$. We will use $\mathcal{A}$ to denote the pair $(\vec{A}^1, \vec{A}^2)$, $\mathcal{B}$ to denote the pair $(\vec{B}^1, \vec{B}^2)$, $\Upsilon^{\alice}$ to denote $(S, \mathcal{A}, \vec{r}^{\alice})$,  $\Upsilon^{\bob}$ to denote $(T, \mathcal{B}, \vec{r}^{\bob})$, and finally $\Upsilon_{-\theta}$ to denote $(\Upsilon^{\alice}, \Upsilon^{\bob}, i_{\star})$. Next, using $\Upsilon$, we define random variables $v^{\alice}_{j}, v^{\bob}_j \in \bxos_m$ for $j \in \{1, 2\}$. To simplify notation, we omit $\Upsilon$ from these random variables even though  they are functions of $\Upsilon$. We define, for $j \in \{1, 2\}$ and $Z \subseteq M$:
\[
v^{\alice}_j(Z) = \max_{F \in \mathcal{F}^{\alice}_j} \lvert{Z \cap F}\rvert \hspace{3cm} v^{\bob}_j(Z) = \max_{F \in \mathcal{F}^{\bob}_j} \lvert{Z \cap F}\rvert, 
\]
where
\[
\mathcal{F}^{\alice}_j =  \{A^{j'}_{i} \mid i \in [n], j' \in [2]\}  \setminus \{A^{3-j}_{i_{\star}}\}  \hspace{1cm} \mathcal{F}^{\bob}_j = \{B^{j'}_{i} \mid i \in [n], j' \in [2]\}  \setminus \{B^{3-j}_{i_{\star}}\} .
\]

Intuitively, $v^\alice_\theta$ has strictly more clauses than $v^\alice$: it contains \emph{every} regular clause (but still only one special clause). While of course Alice does not know the valuation $v^\alice_\theta$ (because she does not know which clause is special), we can still nonetheless use it to upper bound the value of $v^\alice$ for any set. 

\subsection{A Good Allocation Determines $\theta$}
Two key properties establish $\nu$ as a hard distribution. The first property is that $\theta$ can be recovered immediately from any allocation which guarantees a $\nicefrac{3}{4}$-approximation. This is captured in~\autoref{lemma:nu} below.

We mention that the proof of \Cref{item:nu3} of \autoref{lemma:nu} uses the observation that $\lvert{A^j_i}\rvert = \lvert{B^j_i}\rvert = \frac{m}{2}$ for all $i \in [n], j \in [2]$. It also crucially leverages the fact that we are taking the minimum over $j \in \{1, 2\}$ (as is captured by $\forall$). In particular, the same statement with the minimum replaced by an average over $j$ is not true. This should be expected, as otherwise it would contradict the randomized simultaneous algorithm of \cite{BravermanMW18} which guarantees a $3/4$-approximation in expectation. 

In~\Cref{lemma:nu} below,~\Cref{item:nu1} simply states that the optimal welfare is always $m$. We have given intuition for this immediately following the definition of $\nu$, but the proof below makes this rigorous.~\Cref{item:nu2} is straight-forward as $v^\alice_\theta$ has strictly more clauses than $v^\alice$.~\Cref{item:nu3} is the crucial bullet, which states that (except with exponentially small probability) \emph{no allocation achieves welfare $3m/4$ when $\theta =1$ \textbf{and} when $\theta = 2$.} Therefore, learning an allocation which guarantees welfare at least $3m/4$ immediately determines $\theta$.

Recall the definition of $\opt(\cdot)$ from \Cref{sec:caprelim} and that $\Upsilon$ defines $v^{\alice}, v^{\bob}$.

\begin{lemma} \label{lemma:nu} We have:
\begin{enumerate}
\item \label{item:nu1} For all $\Upsilon \sim \nu$, we have $\opt(v^{\alice}, v^{\bob}) = m$.
\item \label{item:nu2} For all $\Upsilon \sim \nu$ and  $Z \subseteq M$, we have $v^{\alice}(Z) \leq v^{\alice}_{\theta}(Z)$ and $ v^{\bob}(Z) \leq v^{\bob}_{\theta}(Z)$.
\item \label{item:nu3} It holds that:
\[
\Pr_{\Upsilon \sim \nu}\left( \exists Z \subseteq M: \forall  j \in \{1, 2\}: v^{\alice}_j(Z) + v^{\bob}_j(\overline{Z}) > \frac{179m}{240} + \varepsilon m  \right) \leq 12n^2 \cdot \exp\left(-\frac{\varepsilon^2m}{20}\right).
\]
\end{enumerate}
\end{lemma}
\begin{proof}
We show each part in turn:
\begin{enumerate}
\item For the first part,  is is enough to show that $\opt(v^{\alice}, v^{\bob}) \geq m$. We have $\opt(v^{\alice}, v^{\bob}) \geq v^{\alice}(A^{\theta}_{i_{\star}}) + v^{\bob}(\overline{A^{\theta}_{i_{\star}}}) = v^{\alice}(A^{\theta}_{i_{\star}}) + v^{\bob}(B^{\theta}_{i_{\star}}) = m$. 
\item For the second part, we only argue for $v^{\alice}(Z) \leq v^{\alice}_{\theta}(Z)$ as the other argument is symmetric. This follows by the definition of $v^{\alice}$ and $v^{\alice}_{\theta}$ and the fact that $\mathcal{F}^{\alice} \subseteq \mathcal{F}^{\alice}_{\theta}$.
\item For the third part, we define the following events over the randomness in $\Upsilon$.
\begin{align*}
E_{reg} &\equiv \exists i, i' \neq i_{\star}, j, j' \in \{1, 2\} : \lvert{A^j_i \cap B^{j'}_{i'}}\rvert < \frac{51m}{200} - \varepsilon m.\\
E^{\alice}_{special} &\equiv \exists i \neq i_{\star}, j \in \{1, 2\} : \lvert{A^j_{i_{\star}} \cap B^{3 -j}_{i}}\rvert < \frac{61m}{240} - \varepsilon m.\\
E^{\bob}_{special} &\equiv \exists i \neq i_{\star}, j \in \{1, 2\} : \lvert{A^{3-j}_{i} \cap B^{j}_{i_{\star}}}\rvert < \frac{61m}{240} - \varepsilon m.
\end{align*}
Finally, define the event $E = E_{reg} \vee E^{\alice}_{special} \vee E^{\bob}_{special}$.  We claim that
\begin{claim*}
 $\Pr(E)  \leq 12n^2 \cdot\exp\left(-\frac{\varepsilon^2m}{20}\right)$. 
 \end{claim*}
 \begin{proof}
 By the union bound, we have $\Pr(E)  \leq \Pr(E_{reg}) + \Pr(E^{\alice}_{special}) + \Pr(E^{\bob}_{special})$. We next show that each one of $\Pr(E_{reg})$, $\Pr(E^{\alice}_{special})$, $\Pr(E^{\bob}_{special})$ is at most $4n^2 \cdot \exp\left(-\frac{\varepsilon^2m}{20}\right)$. 

We start by showing $\Pr(E_{reg}) \leq 4n^2 \cdot \exp\left(-\frac{\varepsilon^2m}{20}\right)$. We derive using \autoref{lemma:reg}:
\[
\Pr(E_{reg}) \leq \sum_{i, i' \neq i_{\star}} \sum_{j, j' \in \{1,2\}} \Pr\left(\lvert{A^j_i \cap B^{j'}_{i'}}\rvert < \frac{51m}{200} - \varepsilon m\right) \leq 4n^2 \cdot \exp\left(-\frac{\varepsilon^2m}{20}\right).
\]
We next show that $\Pr(E^{\alice}_{special}) \leq 4n^2 \cdot \exp\left(-\frac{\varepsilon^2m}{20}\right)$. We derive using \autoref{lemma:special}:
\[
\Pr(E^{\alice}_{special}) \leq \sum_{i \neq i_{\star}} \sum_{j \in \{1,2\}} \Pr\left(\lvert{A^j_{i_{\star}} \cap B^{3-j}_{i}}\rvert < \frac{61m}{240} - \varepsilon m\right) \leq 4n^2 \cdot \exp\left(-\frac{\varepsilon^2m}{20}\right).
\]
Finally, we show that $\Pr(E^{\bob}_{special}) \leq 4n^2 \cdot \exp\left(-\frac{\varepsilon^2m}{20}\right)$. For this part, recall that if a basis $S$ is compatible with $T$, then $T^{rev}$ is compatible with $S^{rev}$. Furthermore, a pair$(A^1_{\star},A^2_{\star})$ is special with respect to $(S,T)$ if and only if $(\overline{A^2_{\star}},\overline{A^1_{\star}})$ is special with respect to $(T^{rev},S^{rev})$. We apply \autoref{lemma:special} on $T^{rev}, S^{rev}$ to get:
\[
\Pr(E^{\bob}_{special}) \leq \sum_{i \neq i_{\star}} \sum_{j \in \{1,2\}} \Pr\left(\lvert{A^{3-j}_{i} \cap B^{j}_{i_{\star}}}\rvert < \frac{61m}{240} - \varepsilon m\right) \leq 4n^2 \cdot \exp\left(-\frac{\varepsilon^2m}{20}\right).
\]

This finishes the proof that $\Pr(E)  \leq 12n^2 \cdot\exp\left(-\frac{\varepsilon^2m}{20}\right)$.
\end{proof} We next claim that whenever we have a $Z \subseteq M$ such that  $ v^{\alice}_j(Z) + v^{\bob}_j(\overline{Z}) > \frac{179m}{240} + \varepsilon m$ for all $ j \in \{1, 2\}$, then $E$ happens. This finishes the proof of the lemma as it follows that:
\begin{align*}
\Pr_{\Upsilon \sim \nu}\left( \exists Z \subseteq M: \forall j \in \{1, 2\}: v^{\alice}_j(Z) + v^{\bob}_j(\overline{Z}) > \frac{179m}{240} + \varepsilon m  \right) &\leq  \Pr(E)  \\
&\leq 12n^2 \cdot \exp\left(-\frac{\varepsilon^2m}{20}\right).
\end{align*}

We now prove the claim. Let $Z \subseteq M$ be such that  $ v^{\alice}_j(Z) + v^{\bob}_j(\overline{Z}) > \frac{179m}{240} + \varepsilon m$ for all $ j \in \{1, 2\}$. Using the definition of $v^{\alice}_j$ and $v^{\bob}_j$, we get that for all $j \in \{1, 2\}$, we have $F^{\alice}_j \in \mathcal{F}^{\alice}_j$ and $F^{\bob}_j \in \mathcal{F}^{\bob}_j$ such that $ \lvert{F^{\alice}_j \cap Z}\rvert + \lvert{F^{\bob}_j \cap \overline{Z}}\rvert > \frac{179m}{240} + \varepsilon m$. We proceed via a case analysis on $F^{\alice}_j, F^{\bob}_j$ for $j \in \{1, 2\}$.

\begin{itemize}
\item {\bf $\bm{\exists j \in [2]: F^{\alice}_j \neq A^j_{i_{\star}} \wedge F^{\bob}_j \neq B^j_{i_{\star}} }$ :} Let $j_{\star}$ be such a $j$. We use the identity $ \lvert{Z' \cap Z}\rvert + \lvert{Z'' \cap \overline{Z}}\rvert \leq  \lvert{Z' \cup Z''}\rvert $ for any sets $Z, Z', Z''$ to get:
\[
\frac{179m}{240} + \varepsilon m < \lvert{F^{\alice}_{j_{\star}} \cap Z}\rvert + \lvert{F^{\bob}_{j_{\star}} \cap \overline{Z}}\rvert  \leq  \lvert{F^{\alice}_{j_{\star}} \cup F^{\bob}_{j_{\star}}}\rvert .
\]
Next, as $F^{\alice}_{j_{\star}} \in \mathcal{F}^{\alice}_{j_{\star}}$ and $F^{\bob}_{j_{\star}} \in  \mathcal{F}^{\bob}_{j_{\star}}$, we have that $\lvert{F^{\alice}_{j_{\star}}}\rvert = \lvert{F^{\bob}_{j_{\star}}}\rvert = \frac{m}{2}$ and we get $\lvert{F^{\alice}_{j_{\star}} \cap F^{\bob}_{j_{\star}}}\rvert < \frac{61m}{240} - \varepsilon m$. As $F^{\alice}_{j_{\star}} \neq A^{j_{\star}}_{i_{\star}}$ and $ F^{\bob}_{j_{\star}} \neq B^{j_{\star}}_{i_{\star}}$, this means that $E_{reg}$ and thus, $E$ happens.

\item {\bf If $\bm{\exists j \in [2]: F^{\alice}_j \in \vec{A}^{3-j} \vee F^{\bob}_j \in \vec{B}^{3-j}  }$ :} Let $j_{\star}$ be such a $j$ and assume that $F^{\alice}_{j_{\star}} \in \vec{A}^{3-j_{\star}}$. The proof is symmetric when $F^{\bob}_{j_{\star}} \in \vec{B}^{3-j_{\star}}$. We begin by showing that  $\vec{A}^{1}$ and $\vec{A}^{2}$ are disjoint. Indeed, all elements of  $\vec{A}^{1}$ are clauses with respect to $S$ whereas all elements of $\vec{A}^{2}$ are clauses with respect to $S^{rev}$ (\autoref{obs:mustar}). By \autoref{def:clause} no set can be a clause with respect to both $S$ and $S^{rev}$ and thus,  $\vec{A}^{1}$ and $\vec{A}^{2}$ must be disjoint.

As $\vec{A}^{1}$ and $\vec{A}^{2}$ are disjoint, we have  that $F^{\alice}_{j_{\star}} \in \vec{A}^{3-j_{\star}} \implies F^{\alice}_{j_{\star}} \notin \vec{A}^{j_{\star}} \implies F^{\alice}_{j_{\star}} \neq  A^{j_{\star}}_{i_{\star}}$. If $F^{\bob}_{j_{\star}} \neq  B^{j_{\star}}_{i_{\star}}$, then we are done by the previous part, so we assume that $F^{\bob}_{j_{\star}} =  B^{j_{\star}}_{i_{\star}}$.

Using the definition of $\mathcal{F}^{\alice}_{j_{\star}}$, we have that $F^{\alice}_{j_{\star}} \notin \vec{A}^{j_{\star}} \implies F^{\alice}_{j_{\star}} = A^{3-j_{\star}}_{i^{\alice}}$ for some $i^{\alice} \neq i_{\star}$.  We use the identity $ \lvert{Z' \cap Z}\rvert + \lvert{Z'' \cap \overline{Z}}\rvert \leq  \lvert{Z' \cup Z''}\rvert $ for any sets $Z, Z', Z''$ to get:
\[
\frac{179m}{240} + \varepsilon m < \lvert{A^{3-j_{\star}}_{i^{\alice}} \cap Z}\rvert + \lvert{B^{j_{\star}}_{i_{\star}} \cap \overline{Z}}\rvert  \leq  \lvert{A^{3-j_{\star}}_{i^{\alice}} \cup B^{j_{\star}}_{i_{\star}}}\rvert .
\]
Next, as  $ \lvert{A^{3-j_{\star}}_{i^{\alice}} }\rvert  =  \lvert{ B^{j_{\star}}_{i_{\star}}}\rvert  = \frac{m}{2}$ and we get $ \lvert{A^{3-j_{\star}}_{i^{\alice}} \cap B^{j_{\star}}_{i_{\star}}}\rvert < \frac{61m}{240} - \varepsilon m$. As $i^{\alice} \neq i_{\star}$, this means that $E^{\bob}_{special}$ and thus, $E$ happens.

\item {\bf Otherwise:} As we are not in case $2$, we can assume that for all $j \in [2]$, we have an $i^{\alice}_j$ and an $i^{\bob}_j$ such that  $F^{\alice}_j = A^{j}_{i^{\alice}_j} $ and $ F^{\bob}_j = B^j_{i^{\bob}_j} $. We have that:
\[
\lvert{A^{1}_{i^{\alice}_1} \cap Z}\rvert + \lvert{B^{1}_{i^{\bob}_1} \cap \overline{Z}}\rvert + \lvert{A^{2}_{i^{\alice}_2} \cap Z}\rvert + \lvert{B^{2}_{i^{\bob}_2} \cap \overline{Z}}\rvert > 2\cdot \left( \frac{179m}{240} + \varepsilon m\right).
\]
By an averaging argument, this means that there exists $j_{\star} \in [2]$ such that $\lvert{A^{j_{\star}}_{i^{\alice}_{j_{\star}}} \cap Z}\rvert +  \lvert{B^{3 - j_{\star}}_{i^{\bob}_{3-j_{\star}}} \cap \overline{Z}}\rvert  >  \frac{179m}{240} + \varepsilon m$. Using $ \lvert{Z' \cap Z}\rvert + \lvert{Z'' \cap \overline{Z}}\rvert \leq  \lvert{Z' \cup Z''}\rvert $ for any sets $Z, Z', Z''$ and the fact that $ \lvert{A^{j_{\star}}_{i^{\alice}_{j_{\star}}} }\rvert  =  \lvert{  B^{3 - j_{\star}}_{i^{\bob}_{3-j_{\star}}} }\rvert  = \frac{m}{2}$ , we get that 
\[
\lvert{A^{j_{\star}}_{i^{\alice}_{j_{\star}}} \cap B^{3 - j_{\star}}_{i^{\bob}_{3-j_{\star}}}}\rvert  <  \frac{61m}{240} - \varepsilon m  .
\]
If $i^{\alice}_{j_{\star}} \neq i_{\star}$ and $i^{\bob}_{3-j_{\star}} \neq i_{\star}$, then the above inequality implies that $E_{reg}$, and therefore $E$ happens. If $i^{\alice}_{j_{\star}} = i_{\star}$ and $i^{\bob}_{3-j_{\star}} \neq i_{\star}$, then the above inequality implies that $E^{\alice}_{special}$, and therefore $E$ happens. If $i^{\alice}_{j_{\star}} \neq i_{\star}$ and $i^{\bob}_{3-j_{\star}} = i_{\star}$, then the above inequality implies that $E^{\bob}_{special}$, and therefore $E$ happens.  Finally, one of these three cases must hold as otherwise, we have  $i^{\alice}_{j_{\star}} = i^{\bob}_{3-j_{\star}} = i_{\star}$, implying
\begin{align*}
\frac{m}{2} - \lvert{A^{1}_{i_{\star}} \cap A^{2}_{i_{\star}}}\rvert  &= \frac{m}{2} - \lvert{A^{j_{\star}}_{i_{\star}} \cap A^{3 - j_{\star}}_{i_{\star}}}\rvert \\
&= \lvert{A^{j_{\star}}_{i_{\star}} \cap B^{3 - j_{\star}}_{i_{\star}}}\rvert <  \frac{61m}{240} - \varepsilon m,
\end{align*}
contradicting \autoref{def:mustar}.

\end{itemize}

\end{enumerate}
\end{proof}

Again, the key aspects of our construction which we have established so far is that (a) the optimal welfare is always $m$, and (b) learning an allocation which achieves welfare $\geq 179m/240+\varepsilon m$ determines $\theta$. Therefore, any algorithm which guarantees a $3/4$-approximation also learns $\theta$. It now remains to show that learning $\theta$ requires exponential communication. 

\subsection{Key Technical Lemma: $i^\star$ is Independent of All Else}
\Cref{sec:lower} contains our final proof that learning $\theta$ requires exponential communication. We wrap up this section with one key lemma regarding our construction. Absent any conditioning, $i^\star$ is clearly a uniformly random index in $[n]$. Clearly, $i^\star$ is not uniformly random conditioned on the entire rest of the construction (because it is the only index with a special clause, which can be determined from the rest of the construction). However, we have carefully constructed $\nu$ so that $i^\star$ remains a uniformly random index in $[n]$, \emph{even conditioning on Alice's other information} (and ditto for Bob). \Cref{lemma:istar} states this formally.

\begin{lemma}
\label{lemma:istar} For the random variable $ \Upsilon = (\Upsilon^{\alice}, \Upsilon^{\bob}, i_{\star}, \theta)$, it holds that:
\begin{enumerate}
\item The marginal $i_{\star}$ is independent of the marginal $\Upsilon^{\alice}$.
\item The marginal $i_{\star}$ is independent of the marginal $\Upsilon^{\bob}$.
\end{enumerate}
\end{lemma}
\begin{proof}
We only show the first claim as the second one is similar. To show that the marginal $i_{\star}$ is independent of the marginal $\Upsilon^{\alice}$, we show that the distribution $\nu$ is equivalent to the distribution $\nu'$ below. It is clear from the definition of $\nu'$ that the marginal $i_{\star}$ is independent of the marginal $\Upsilon^{\alice}$.

\begin{tbox}
\begin{itemize}[leftmargin=10pt]
	\item \textbf{Sampling $(v^{\alice} , v^{\bob}) \sim \nu'$}:  Recall $n = \exp\left(\frac{\eps^2 \cdot m}{100}\right)$.
\begin{enumerate}[label=(\arabic*)]
	\item Sample a basis $S \sim \xi_{single}$.
	\item Construct sequences $\vec{A}^1, \vec{A}^2$ of $n$ subsets of $M$  (where $\vec{A}^1 = A^1_1, \cdots, A^1_n$, {\em etc.}) by  sampling $(A^1_i, A^2_i)  \sim \mu(S)$ independently for $ i \in [n]$. 
	\item \label{line:nufinalbasis} Sample $i_{\star} \sim \unif([n])$ and let $T$ be sampled uniformly at random such that $\lvert{\part_{S \Vert T \Vert A^1_{i_{\star}} \Vert A^2_{i_{\star}}}}\rvert  = \vecopt$. Observe that any such $T$ is a basis. We show in our proof that at least one such $T$  exists and therefore this step is well defined. 
	\item Construct sequences $\vec{B}^1, \vec{B}^2$ of $n$ subsets of $M$ (where $\vec{B}^1 = B^1_1, \cdots, B^1_n$, {\em etc.}) as follows:
	\begin{enumerate}[label=(\alph*)]
		\item For $i \neq i_{\star} \in [n]$,  sample $(B^2_i, B^1_i)  \sim \mu(T^{rev})$ independently. 
		\item Set $(B^1_{i_{\star}}, B^2_{i_{\star}})  =  (\overline{A^1_{i_{\star}}}, \overline{A^2_{i_{\star}}})$. 
	\end{enumerate}
	\item Sample $\theta \in \unif(\{1,2\})$, and sequences $\vec{r}^{\alice}= r^{\alice}_1,\cdots,r^{\alice}_n \in \{1,2\}^{n}$ and $\vec{r}^{\bob}= r^{\bob}_1,\cdots,r^{\bob}_n \in \{1,2\}^{n}$ uniformly at random subject to $r^{\alice}_{i_{\star}} = r^{\bob}_{i_{\star}} = \theta$.
	\item Define $v^{\alice}(Z) = \max_{F \in \mathcal{F}^{\alice}} \vert{Z \cap F}\rvert$  and $v^{\bob}(Z) = \max_{F \in \mathcal{F}^{\bob}} \vert{Z \cap F}\rvert$ where for all $Z \subseteq M$:
	\[
		\mathcal{F}^{\alice} = \{A^{r^{\alice}_i}_i \mid i \in [n]\} \hspace{1cm} \text{ and }\hspace{1cm} \mathcal{F}^{\bob} = \{B^{r^{\bob}_i}_i \mid i \in [n]\}.
	\]
\end{enumerate}
\end{itemize}
\end{tbox}

We first show why \autoref{line:nufinalbasis} in the definition of $\nu'$ is well defined. For this, we need to show that for any basis $S$ and any $(A^1, A^2)$ in the support of $\mu(S)$, there exists a $T$ such that $\lvert{\part_{S \Vert T \Vert A^1 \Vert A^2}}\rvert  = \vecopt$. As for any basis $S$ and all $(A^1, A^2)$ in the support of $\mu(S)$, the value of $\lvert{\part_{S \Vert A^1 \Vert A^2}}\rvert$ (\autoref{obs:mu}) is the same, by symmetry, it is sufficient to show this for any one $(A^1, A^2)$ in the support of $\mu(S)$ for any one $S$. But such an $S, (A^1, A^2)$ and $T$ is described in \autoref{fig:compat}

Next, we show why distribution $\nu$ is equivalent to distribution $\nu'$, proceeding in steps, each time changing the description of $\nu$ a little bit so that it eventually becomes $\nu'$. We show that the distributions described in all the steps are equivalent.

\begin{itemize}
\item{\bf Step (a):} In this step, we replace Line~\ref{line:nubases} in the definition of $\nu$ by the following:
\begin{quote}
\begin{enumerate}[label={(\arabic*a)}]
\item \label{line:nu1bases} Sample a basis $S \sim \xi_{single}$ and  basis $T$ uniformly at random such that $S$ is compatible with $T$. This step is well defined for the same reason as above.
\end{enumerate}
\end{quote}
To show that this does not affect the actual distribution, we use \autoref{lemma:part}. We get that, for all bases $Z, Z' $, 
\begin{align*}
\Pr_{(S, T) \sim \xi}\left((S, T) = (Z, Z')\right) &= \Pr_{(S, T) \sim \xi}\left(S = Z\right)\Pr_{(S, T) \sim \xi}\left(T = Z' \mid S = Z \right)\\
&= \Pr_{\substack{S \sim \xi_{single} \\ T \sim \xi_{single}}}\left(S = Z \mid\lvert{\part_{S \Vert T}}\rvert = \veccmp \right)\Pr_{(S, T) \sim \xi}\left(T = Z' \mid S = Z \right) \\
&= \Pr_{S \sim \xi_{single}}\left(S = Z\right)\Pr_{(S, T) \sim \xi}\left(T = Z' \mid S = Z \right)\tag{\autoref{lemma:part}}\\
&= \Pr_{S \sim \xi_{single}}\left(S = Z\right)\Pr_{T \sim \xi_{single}}\left(T = Z' \mid \text{$Z$ is compatible with $T$} \right)\tag{Definition of $\xi$},
\end{align*}
as required.

\item{\bf Step (b):} In this step, we replace and Line~\ref{line:nu1bases} from Step (a) and Line~\ref{line:nuclauses} in the definition of $\nu$ by the following:
\begin{quote}
\begin{enumerate}[label={(\arabic*b)}]
\item Sample a basis $S \sim \xi_{single}$.
\item \label{line:nu2clauses} Sample $i_{\star} \sim \unif([n])$ and construct sequences $\vec{A}^1, \vec{A}^2, \vec{B}^1, \vec{B}^2$ of $n$ subsets of $M$ as follows (where $\vec{A}^1 = A^1_1, \cdots, A^1_n$, {\em etc.}):
	\begin{enumerate}[label=(\alph*)]
		\item For $i \neq i_{\star} \in [n]$,  sample $(A^1_i, A^2_i)  \sim \mu(S)$ independently.
		\item Sample basis $T$ uniformly at random such that $S$ is compatible with~$T$. 
		\item Sample $(A^1_{\star}, A^2_{\star}) \sim \mu_{\star}(S, T)$ and set $(A^1_{i_{\star}}, A^2_{i_{\star}})  =  (A^1_{\star}, A^2_{\star})$. 
		\item For $i \neq i_{\star} \in [n]$,  sample  $(B^2_i, B^1_i)  \sim \mu(T^{rev})$ independently. 
		\item Set $(B^1_{i_{\star}}, B^2_{i_{\star}})  =  (\overline{A^1_{i_{\star}}}, \overline{A^2_{i_{\star}}})$. 
	\end{enumerate}
\end{enumerate}
\end{quote}
This change does not affect the distribution as $i_{\star}$ and $(A^1_i, A^2_i)$ for $i \neq i_{\star}$ were picked independently of $T$ and $(B^2_i, B^1_i)$ for $i \neq i_{\star}$ were picked independently of $(A^1_{\star}, A^2_{\star})$, and thus we can interchange the order in which these are picked.

\item{\bf Step (c):} In this step, we replace and Line~\ref{line:nu2clauses} from Step (b) by the following:
\begin{quote}
\begin{enumerate}[label={(\arabic*c)}]
\setcounter{enumi}{1}
\item Sample $i_{\star} \sim \unif([n])$ and construct sequences $\vec{A}^1, \vec{A}^2, \vec{B}^1, \vec{B}^2$ of $n$ subsets of $M$ as follows (where $\vec{A}^1 = A^1_1, \cdots, A^1_n$, {\em etc.}):
	\begin{enumerate}[label=(\alph*)]
		\item \label{line:nu3aclause}For $i \in [n]$,  sample $(A^1_i, A^2_i)  \sim \mu(S)$ independently.
		\item Sample basis $T$ uniformly at random such that $\lvert{\part_{S \Vert T \Vert A^1_{i_{\star}} \Vert A^2_{i_{\star}}}}\rvert  = \vecopt$. Observe that any such $T$ is always a basis.
		\item For $i \neq i_{\star} \in [n]$,  sample  $(B^2_i, B^1_i)  \sim \mu(T^{rev})$ independently. 
		\item Set $(B^1_{i_{\star}}, B^2_{i_{\star}})  =  (\overline{A^1_{i_{\star}}}, \overline{A^2_{i_{\star}}})$.
	\end{enumerate}
\end{enumerate}
\end{quote}

Before showing that this change does not affect the distribution, we define some helpful notation. For a basis $S$, we let $\xi_{cmp}(S)$ denote that the uniform distribution over all bases $T$ such that $S$ is compatible with $T$.  Using this notation, we get that for all bases $Z$ and $ Z^1, Z^2 \subseteq M$: 
\begin{align*}
&\Pr_{T \sim \xi_{cmp}(S)}\left(T = Z \right) \Pr_{(A^1_{\star}, A^2_{\star}) \sim \mu_{\star}(S, Z)}\left((A^1_{\star}, A^2_{\star}) = (Z^1, Z^2)\right) \\
&\hspace{0.5cm}= \Pr_{T \sim \xi_{cmp}(S)}\left(T = Z \right) \Pr_{(A^1, A^2) \sim \mu(S)}\left((A^1, A^2) = (Z^1, Z^2) \mid \lvert{\part_{S \Vert Z \Vert A^1 \Vert A^2}}\rvert = \vecopt\right) \tag{\autoref{obs:mustar}}\\
&\hspace{0.5cm}= \Pr_{T \sim \xi_{cmp}(S)}\left(T = Z \right) \Pr_{\substack{T \sim \xi_{cmp}(S) \\ (A^1, A^2) \sim \mu(S)}}\left((A^1, A^2) = (Z^1, Z^2) \mid \lvert{\part_{S \Vert T \Vert A^1 \Vert A^2}}\rvert = \vecopt, T = Z\right) \\
&\hspace{0.5cm}= \Pr_{\substack{T \sim \xi_{cmp}(S) \\ (A^1, A^2) \sim \mu(S)}}\left((T, A^1, A^2) = (Z, Z^1, Z^2) \mid \lvert{\part_{S \Vert T \Vert A^1 \Vert A^2}}\rvert = \vecopt\right)  \tag{\autoref{obs:mu}, \autoref{lemma:part}}\\
&\hspace{0.5cm}= \Pr_{(A^1, A^2) \sim \mu(S)} \left((A^1, A^2) = (Z^1, Z^2) \right)\\
&\hspace{3cm} \times \Pr_{\substack{T \sim \xi_{cmp}(S) \\ (A^1, A^2) \sim \mu(S)}}\left(T = Z \mid \lvert{\part_{S \Vert T \Vert A^1 \Vert A^2}}\rvert = \vecopt, (A^1, A^2) = (Z^1, Z^2) \right)  \tag{\autoref{obs:mu}, \autoref{lemma:part}} \\
&\hspace{0.5cm}= \Pr_{(A^1, A^2) \sim \mu(S)} \left((A^1, A^2) = (Z^1, Z^2) \right) \times \Pr_{T \sim \xi_{cmp}(S) }\left(T = Z \mid \lvert{\part_{S \Vert T \Vert Z^1 \Vert Z^2}}\rvert = \vecopt \right)  ,
\end{align*}
as desired.

\item{\bf Step (d):}
To finish the proof, we claim that $\nu'$ is the same as the distribution in Step (c) above. This is because $(A^1_i, A^2_i)$ for $i \in [n]$ were picked independently of $i_{\star}$ in Line~\ref{line:nu3aclause} of the distribution in Step~(c) and thus we can interchange the order in which they are picked. As interchanging this order converts the distribution in Step (c) above to $\nu'$, we are done.

\end{itemize}

\end{proof}

\section{The Proof of \autoref{thm:mainred}}\label{sec:lower}

In this section, we complete our proof of \autoref{thm:mainred}. Our proof crucially relies on \autoref{lemma:nu} and \autoref{lemma:istar} from \Cref{sec:dist}. Note that the  remaining task is to establish that exponential communication is required to learn non-trivial information about $\theta$.

\begin{proof}[Proof of \autoref{thm:mainred}]
Let $\varepsilon >0$ and $m > \frac{10^{10}}{\varepsilon^2}$ be arbitrary. By Yao's minimax principle, in order to show \autoref{thm:mainred}, it is sufficient to show a distribution $\nu$ over pairs of functions from $\bxos_m$ such that any {\em deterministic} combinatorial auction that is simultaneous and $\left(\frac{3}{4} - \frac{1}{240} + \varepsilon \right)$-approximate over $\nu$ with probability $\frac{1}{2} + \exp\left( - \frac{ \varepsilon^2 m }{500}\right)$ satisfies $\cc(\Pi) \geq \exp\left(\frac{ \varepsilon^2 m }{500}\right) $.

We let $\nu$ denote the distribution defined in \Cref{sec:nu} for $m, \varepsilon$ and let $\Upsilon$ be a random variable denoting a sample from $\nu$ as in \Cref{sec:nu}. Recall how $\Upsilon$ defines the valuation functions $v^{\alice}$, $v^{\bob}$, and also $v^{\alice}_j, v^{\bob}_j$ for $j \in [2]$. Fix $\Pi$ to be a simultaneous deterministic mechanism that is $\left(\frac{3}{4} - \frac{1}{240} + \varepsilon \right)$-approximate over $\nu$ with probability $\frac{1}{2} + \exp\left( - \frac{ \varepsilon^2 m }{500}\right)$. We have from \Cref{sec:caprelim} that 
\begin{equation}
\label{eq:mainred1}
\Pr_{\Upsilon \sim \nu}\left(v^{\alice}(\alloc^{\alice}_{\Pi}(v^{\alice}, v^{\bob})) + v^{\bob}(\alloc^{\bob}_{\Pi}(v^{\alice}, v^{\bob})) > \Big(\frac{179}{240} + \varepsilon \Big) \cdot \opt(v^{\alice}, v^{\bob}) \right) \geq \frac{1}{2} + \exp\Big( - \frac{ \varepsilon^2 m }{500}\Big).
\end{equation}
To simplify notation, we will henceforth omit $\Upsilon \sim \nu$ with the understanding that all the probabilities and expectations are over the randomness in $\Upsilon \sim \nu$.
We use \Cref{item:nu1} and \Cref{item:nu2} of \autoref{lemma:nu}, the fact that the functions $v^{\alice}$ and $v^{\bob}$ are monotone, and that $\alloc^{\alice}_{\Pi}(v^{\alice}, v^{\bob})$ and $\alloc^{\bob}_{\Pi}(v^{\alice}, v^{\bob})$ are disjoint to get the following from \autoref{eq:mainred1}:
\begin{equation}
\label{eq:mainred2}
\Pr\left(v^{\alice}_{\theta}(Z(\Upsilon)) + v^{\bob}_{\theta}(\overline{Z(\Upsilon)}) > \Big(\frac{179}{240} + \varepsilon \Big) \cdot m \right) \geq \frac{1}{2} + \exp\Big( - \frac{ \varepsilon^2 m }{500}\Big),
\end{equation}
where $Z(\Upsilon) = \alloc^{\alice}_{\Pi}(v^{\alice}, v^{\bob})$. Let 
\[
E_{bad} = \exists Z \subseteq M: \forall  j \in \{1, 2\}: v^{\alice}_j(Z) + v^{\bob}_j(\overline{Z}) > \Big(\frac{179}{240} + \varepsilon\Big) m,
\]
be the event from \Cref{item:nu3} of \autoref{lemma:nu}. By the law to total probability we have 
\begin{equation}
\label{eq:mainred3}
\begin{split}
&\Pr\left(v^{\alice}_{\theta}(Z(\Upsilon)) + v^{\bob}_{\theta}(\overline{Z(\Upsilon)}) > \Big(\frac{179}{240} + \varepsilon \Big) \cdot m \right)\\
&\hspace{2cm}\leq \Pr\left(E_{bad} \right) + \Pr\left(\overline{E_{bad}}  \wedge v^{\alice}_{\theta}(Z(\Upsilon)) + v^{\bob}_{\theta}(\overline{Z(\Upsilon)}) > \Big(\frac{179}{240} + \varepsilon \Big) \cdot m \right)\\
&\hspace{2cm}\leq 12n^2 \cdot \exp\left(-\frac{\varepsilon^2m}{20}\right) + \Pr\left(\overline{E_{bad}}  \wedge v^{\alice}_{\theta}(Z(\Upsilon)) + v^{\bob}_{\theta}(\overline{Z(\Upsilon)}) > \Big(\frac{179}{240} + \varepsilon \Big) \cdot m \right)\\
&\hspace{2cm}\leq 12n^2 \cdot \exp\left(-\frac{\varepsilon^2m}{20}\right) \\
&\hspace{3cm} + \Pr\left( v^{\alice}_{\theta}(Z(\Upsilon)) + v^{\bob}_{\theta}(\overline{Z(\Upsilon)}) > v^{\alice}_{3- \theta}(Z(\Upsilon)) + v^{\bob}_{3- \theta}(\overline{Z(\Upsilon)}) \right),
\end{split}
\end{equation}
using \Cref{item:nu3} of \autoref{lemma:nu} in the penultimate step. Now, we focus on the second term in the expression above. For every value $\omega$ that the tuple $(\mathcal{A},  \mathcal{B}, i_{\star})$ can take, we define the event $E_{\omega} \equiv (\mathcal{A},  \mathcal{B}, i_{\star}) = \omega$. By the law of total probability, we have 
\begin{align*}
&\Pr\left( v^{\alice}_{\theta}(Z(\Upsilon)) + v^{\bob}_{\theta}(\overline{Z(\Upsilon)}) > v^{\alice}_{3- \theta}(Z(\Upsilon)) + v^{\bob}_{3- \theta}(\overline{Z(\Upsilon)}) \right) \\
&\hspace{1cm}\leq \sum_{\omega} \sum_{Z \subseteq [m]} \sum_{j \in [2]} \Pr(E_{\omega}  \wedge Z(\Upsilon) = Z )\Pr(\theta = j \mid  E_{\omega}, Z(\Upsilon) = Z )\\
&\hspace{2cm}\times \Pr\left( v^{\alice}_{\theta}(Z(\Upsilon)) + v^{\bob}_{\theta}(\overline{Z(\Upsilon)}) > v^{\alice}_{3- \theta}(Z(\Upsilon)) + v^{\bob}_{3- \theta}(\overline{Z(\Upsilon)}) \mid E_{\omega},  Z(\Upsilon) = Z , \theta = j  \right).
\end{align*}
Observe that conditioning on $E_{\omega},  Z(\Upsilon) = Z$ fixes the value of $v^{\alice}_{1}(Z(\Upsilon)) + v^{\bob}_{1}(\overline{Z(\Upsilon)})$ and $ v^{\alice}_{2}(Z(\Upsilon)) + v^{\bob}_{2}(\overline{Z(\Upsilon)})$. Thus, the last factor in the summand above is either $0$ or $1$ and it can be $1$ for at most one value of $\theta$. We conclude:
\begin{equation}
\label{eq:mainred4}
\begin{split}
&\Pr\left( v^{\alice}_{\theta}(Z(\Upsilon)) + v^{\bob}_{\theta}(\overline{Z(\Upsilon)}) > v^{\alice}_{3- \theta}(Z(\Upsilon)) + v^{\bob}_{3- \theta}(\overline{Z(\Upsilon)}) \right) \\
&\hspace{1cm}\leq \sum_{\omega} \sum_{Z \subseteq [m]}  \Pr(E_{\omega}  \wedge Z(\Upsilon) = Z ) \max_{j \in [2]}\Pr(\theta = j \mid  E_{\omega}, Z(\Upsilon) = Z ).
\end{split}
\end{equation}
Next, we concentrate on upper bounding the term $\max_{j \in [2]}  \Pr(\theta = j \mid  E_{\omega}, Z(\Upsilon) = Z )$. Since $\theta$ is chosen independently of $\mathcal{A}, \mathcal{B}, i_{\star}$ in the distribution $\nu$, we have 
\begin{align*}
\max_{j \in [2]}  \Pr(\theta = j \mid  E_{\omega}, Z(\Upsilon) = Z ) &= \frac{1}{2} + \max_{j \in [2]}  \Big( \Pr(\theta = j \mid  E_{\omega}, Z(\Upsilon) = Z ) - \frac{1}{2} \Big) \\
&= \frac{1}{2} + \max_{j \in [2]}\Big(  \Pr(\theta = j \mid  E_{\omega}, Z(\Upsilon) = Z ) -  \Pr(\theta = j \mid E_{\omega}) \Big) \\
&= \frac{1}{2} + \tvd{\distribution{\theta \mid E_{\omega}, Z(\Upsilon) = Z }}{\distribution{\theta \mid E_{\omega}}}\tag{\autoref{def:tvd}}\\
&\leq  \frac{1}{2} + \sqrt{\frac{1}{2} \cdot \kl{\distribution{\theta \mid E_{\omega}, Z(\Upsilon) = Z }}{\distribution{\theta \mid E_{\omega}}}}\tag{\autoref{fact:kltvd}, \Cref{item:fact:kltvd2}}\\
\end{align*}
Plugging into \autoref{eq:mainred3} and \autoref{eq:mainred4} and using concavity of $\sqrt{\cdot}$, we get
\begin{equation}
\label{eq:mainred5}
\begin{split}
&\Pr\left(v^{\alice}_{\theta}(Z(\Upsilon)) + v^{\bob}_{\theta}(\overline{Z(\Upsilon)}) > \Big(\frac{179}{240} + \varepsilon \Big) \cdot m \right) \\
&\hspace{0.5cm}\leq \frac{1}{2} + 12n^2 \cdot \exp\left(-\frac{\varepsilon^2m}{20}\right) \\
&\hspace{1cm}+ \sqrt{ \frac{1}{2} \cdot \sum_{\omega} \sum_{Z \subseteq [m]} \Pr(E_{\omega}  \wedge Z(\Upsilon) = Z )  \kl{\distribution{\theta \mid E_{\omega}, Z(\Upsilon) = Z }}{\distribution{\theta \mid E_{\omega}}}}\\
&\hspace{0.5cm}\leq \frac{1}{2} + 12n^2 \cdot \exp\left(-\frac{\varepsilon^2m}{20}\right) +  \sqrt{ \frac{1}{2} \cdot \I(\theta; Z(\Upsilon) \mid \mathcal{A},  \mathcal{B}, i_{\star})} .
\end{split}
\end{equation}
To finish the proof, we claim that 
 \begin{lemma}
 \label{lemma:infotheta}
 It holds that $ \I(\theta; Z(\Upsilon) \mid \mathcal{A},  \mathcal{B}, i_{\star}) \leq 4\cdot\frac{\cc(\Pi)}{n}$.
 \end{lemma}
 We prove \autoref{lemma:infotheta}  later but assuming it for now, we can combine \autoref{eq:mainred2} and \autoref{eq:mainred5} as
 \[
 \exp\Big( - \frac{ \varepsilon^2 m }{500}\Big) \leq 12n^2 \cdot \exp\left(-\frac{\varepsilon^2m}{20}\right) +  \sqrt{2 \cdot  \frac{\cc(\Pi)}{n}},
 \]
and \autoref{thm:mainred} follows using $n =  \exp\Big(\frac{ \varepsilon^2 m }{100}\Big)$.
\end{proof}

We finish this section by showing \autoref{lemma:infotheta}.
\begin{proof}[Proof of \autoref{lemma:infotheta}]
Let $\Pi^{\alice}$ and $\Pi^{\bob}$ be random variables denoting the message sent by Alice and Bob to the Seller in the first round of $\Pi$ when inputs to Alice and Bob are drawn from the distribution $\nu$. As $\Pi$ is simultaneous, it has only one round and $Z(\Upsilon)$ is a function of $\Pi^{\alice}$ and $\Pi^{\bob}$. We get, invoking \autoref{lemma:info} multiple times:
\begin{align*}
\I(\theta; Z(\Upsilon) \mid \mathcal{A},  \mathcal{B}, i_{\star})  &\leq \I(\theta; \Pi^{\alice} \Pi^{\bob} \mid \mathcal{A},  \mathcal{B}, i_{\star}) \tag{\Cref{item:fact:info5} of \autoref{fact:info}}\\
&= \I(\theta; \Pi^{\alice} \mid \mathcal{A},  \mathcal{B}, i_{\star})  + \I(\theta; \Pi^{\bob}  \mid \mathcal{A},  \mathcal{B}, i_{\star}, \Pi^{\alice}) \tag{\Cref{item:fact:info4} of \autoref{fact:info}}\\
&\leq \I(\theta; \Pi^{\alice} \mid \mathcal{A},  \mathcal{B}, i_{\star})  + \I(\theta; \Pi^{\bob}  \mid \mathcal{A},  \mathcal{B}, i_{\star}) + \I(\Pi^{\alice} ; \Pi^{\bob}  \mid \mathcal{A},  \mathcal{B}, i_{\star}, \theta) \\
&\leq \I(\theta; \Pi^{\alice} \mid \mathcal{A}, i_{\star})  +  \I(\theta; \Pi^{\bob} \mid  \mathcal{B}, i_{\star})  \\
&\hspace{1cm}+ \I( \mathcal{B}; \Pi^{\alice}  \mid \mathcal{A}, i_{\star}, \theta ) + \I( \mathcal{A}; \Pi^{\bob}  \mid \mathcal{B}, i_{\star}, \theta) + \I(\Pi^{\alice} ; \Pi^{\bob}  \mid \mathcal{A},  \mathcal{B}, i_{\star}, \theta)
\end{align*}

We now show that the last $3$ terms are all $0$. To show this, we go term by term  using the fact that $\Pi^{\alice}$ is a function of Alice's input $v^{\alice}$, and therefore a function of $\mathcal{A}, \vec{r}^{\alice}$. Similarly, $\Pi^{\bob}$ is a function of Bob's input $v^{\bob}$, and therefore a function of $\mathcal{B}, \vec{r}^{\bob}$. For the term $\I( \mathcal{B}; \Pi^{\alice}  \mid \mathcal{A}, i_{\star}, \theta )$, we get $\I( \mathcal{B}; \Pi^{\alice}  \mid \mathcal{A}, i_{\star}, \theta ) \leq \I( \mathcal{B}; \mathcal{A} \vec{r}^{\alice}  \mid \mathcal{A}, i_{\star}, \theta ) = \I( \mathcal{B}; \vec{r}^{\alice}_{- i_{\star}}  \mid \mathcal{A}, i_{\star}, \theta ) = 0$ as $\theta = r^{\alice}_{i_{\star}}$ and $\vec{r}^{\alice}_{- i_{\star}}$ is sampled independently of $\mathcal{A}, \mathcal{B}, i_{\star}, \theta$.  Recall that $\vec{r}^{\alice}_{- i_{\star}}$ denotes $\vec{r}^{\alice}$ with the coordinate $i_{\star}$ removed. Similarly, we can deduce that $\I( \mathcal{A}; \Pi^{\bob}  \mid \mathcal{B}, i_{\star}, \theta) = 0 $. Finally, for the term $\I(\Pi^{\alice} ; \Pi^{\bob}  \mid \mathcal{A},  \mathcal{B}, i_{\star}, \theta)$, we get $\I(\Pi^{\alice} ; \Pi^{\bob}  \mid \mathcal{A},  \mathcal{B}, i_{\star}, \theta) \leq \I(\mathcal{A} \vec{r}^{\alice} ; \mathcal{B} \vec{r}^{\bob}  \mid \mathcal{A},  \mathcal{B}, i_{\star}, \theta) =  \I( \vec{r}^{\alice}_{-i_{\star}} ;  \vec{r}^{\bob}_{-i_{\star}} \mid \mathcal{A},  \mathcal{B}, i_{\star}, \theta)  =0$ as $\vec{r}^{\alice}_{- i_{\star}}$ is sampled independently of $\vec{r}^{\bob}_{- i_{\star}}, \mathcal{A}, \mathcal{B}, i_{\star}, \theta$. 
Combining, we get
\[
\I(\theta; Z(\Upsilon) \mid \mathcal{A},  \mathcal{B}, i_{\star})  \leq \I(\theta; \Pi^{\alice} \mid \mathcal{A}, i_{\star})  +  \I(\theta; \Pi^{\bob} \mid  \mathcal{B}, i_{\star}) .
\]
We next show that $\I(\theta; \Pi^{\alice} \mid \mathcal{A}, i_{\star})   \leq 2 \cdot \frac{\cc(\Pi)}{n}$. A similar argument shows that $\I(\theta; \Pi^{\bob} \mid  \mathcal{B}, i_{\star}) \leq 2 \cdot \frac{\cc(\Pi)}{n}$ finishing the proof of \autoref{lemma:infotheta}. As $\theta = r^{\alice}_{i_{\star}}$, $\Pi^{\alice}$ is a function of $\mathcal{A}$ and $ \vec{r}^{\alice}$,  and $i_{\star}$ is sampled from $\unif([n])$, we have by \autoref{lemma:istar},
\begin{align*}
\I(\theta; \Pi^{\alice} \mid \mathcal{A}, i_{\star}) &= \I( r^{\alice}_{i_{\star}} ; \Pi^{\alice} \mid \mathcal{A}, i_{\star})\\
&\leq \frac{1}{n} \cdot \I( r^{\alice} ; \Pi^{\alice} \mid \mathcal{A}) \tag{\autoref{lemma:indexinfo}}\\
&\leq \frac{1}{n} \cdot \H(\Pi^{\alice}) \leq \frac{\cc(\Pi) + 1}{n} \leq 2\cdot \frac{\cc(\Pi)}{n}.
\end{align*}

We note that we lose an extra `$+1$' in the argument only because, in our model in \Cref{sec:caprelim}, the length of Alice's and Bob's messages can be anywhere from $0$ to $\cc(\Pi)$. Thus, the total number of possible messages can be upper bounded by $2^{\cc(\Pi) + 1}$ but not $2^{\cc(\Pi)}$.

\end{proof}

\bibliographystyle{alpha}
\bibliography{../MasterBib.bib}

\clearpage
\appendix

\section{Tools from Information Theory}\label{sec:info}

We include a very brief summary of the tools from information theory that we use in this paper. We refer the interested reader to the textbook by Cover and Thomas~\cite{CoverT06} for an excellent introduction to this field. 

\subsection{Entropy and Mutual Information}
\begin{definition}[Entropy]
The Shannon Entropy of a discrete random variable $X$ is defined as 
\[\H(X) = \sum_{x \in \supp(X)} \Pr(X = x) \log \frac{1}{\Pr(X = x)},\]
where $\supp(X)$ is the set of all values $X$ can take and $0 \log \frac{1}{0} = 0$ by convention.
\end{definition}

\begin{definition}[Conditional Entropy] 
\label{def:cond-ent}
Let $X$ and $Y$ be discrete random variables. The entropy of  $X$ conditioned on $Y$ is defined as 
\[\H(X \mid Y) = \E_{y \sim \distribution{Y}} \left[\H(X \mid Y = y)\right].\]
\end{definition}

\begin{definition}[Mutual Information] \label{def:mi} Let $X$, $Y$, and $Z$ be discrete random variables. The mutual information between $X$ and $Y$ is defined as 
\[\I(X; Y) = \H(X) - \H(X\mid Y).\]
The  conditional mutual information  between $X$ and $Y$ conditioned on $Z$ is defined as:
\[\I(X; Y \mid Z) = \H(X\mid Z) - \H(X\mid YZ).\]
\end{definition}

We note that mutual information is symmetric in $X$ and $Y$, {\em i.e.} $\I(Y; X \mid Z) = \I(X; Y \mid Z)$ and $\I(X; Y) = \I(Y; X)$.

\begin{fact} \label{fact:info} The following holds for discrete random variables $W, X, Y, Z$:
\begin{enumerate}
\item \label{item:fact:info1} We have $\H(XY) = \H(X) + \H(Y \mid X) \leq \H(X) + \H(Y)$. Equality holds if $X$ and~$Y$ are independent.
\item  \label{item:fact:info2}  If the random variable $X$ takes values in the set $\Omega$, it holds that $0 \leq \H(X)\leq \log\lvert{\Omega}\rvert$.
\item  \label{item:fact:info3} We have $0 \leq \I(X ; Y\mid Z)\leq \H(X)$ and $\I(X ; Y\mid Z) = 0$ if and only if $X$ is independent of $Y$ given $Z$.
\item  \label{item:fact:info4}  Chain rule of mutual information: 
\[\I(WX ; Y \mid Z) =  \I(W; Y \mid Z)+ \I(X ; Y \mid WZ).\]
\item  \label{item:fact:info5}  Data processing inequality: for any deterministic function $f$, 
\[
	\I(X;f(Y) \mid Z) \leq \I(X;Y \mid Z). 
\]
\end{enumerate}
\end{fact}

We also use the following technical lemmas about mutual information.

\begin{lemma} 
\label{lemma:info}
For discrete random variables $W$, $X$, $Y$, and $Z$, we have 
\[\max(\I(W; X \mid YZ), \I(Y; X \mid Z)) \leq \I(W; X \mid Z) + \I(Y; X \mid WZ).\]
\end{lemma}
\begin{proof} Observe that:
\begin{align*}
\max(\I(W; X \mid YZ), \I(Y; X \mid Z)) &\leq \I(W; X \mid YZ) +  \I(Y; X \mid Z)  \tag{\Cref{item:fact:info3}, \autoref{fact:info}}\\
&= \I(WY; X \mid Z)  \tag{\Cref{item:fact:info4}, \autoref{fact:info}}\\
&= \I(W; X \mid Z) + \I(Y; X \mid WZ).  \tag{\Cref{item:fact:info4}, \autoref{fact:info}}
\end{align*}
\end{proof}

\begin{lemma} 
\label{lemma:indexinfo}
Let $n > 0$ and $X = X_1, X_2, \cdots, X_n$ where $ X_1, X_2, \cdots, X_n$ are independent and identically distributed discrete random variables. Let $I$ be a random variable distributed uniformly over $[n]$. For all discrete random variables $Y$ such that $X$ is independent of $Y$ and $I$ is independent of $(X, Y)$ and all functions $f$, we have:
\[\I(X_I ; f(X, Y) \mid Y, I) \leq \frac{1}{n} \cdot \I(X; f(X, Y) \mid Y).\] 
\end{lemma}
\begin{proof}
Using the fact that $I$ is distributed uniformly over $[n]$, we get
\begin{align*}
\I(X_I ; f(X, Y) \mid Y, I)  &= \H( f(X, Y) \mid Y, I)  - \H(f(X, Y)\mid X_I , Y, I) \tag{\autoref{def:mi}}\\
&= \frac{1}{n} \cdot \sum_{i \in [n]} \left( \E_{y \sim \distribution{Y}}\left[ \H(f(X, Y) \mid Y = y, I = i)\right] \right. \\ 
&\hspace{1.5cm} - \left. \E_{y \sim \distribution{Y}}\E_{x \sim \distribution{X_i} }\left[ \H(f(X, Y) \mid X_i = x, Y = y, I = i)\right]\right) \tag{\autoref{def:cond-ent}}\\ 
&= \frac{1}{n} \cdot \sum_{i \in [n]} \left( \E_{y \sim \distribution{Y}}\left[ \H(f(X, Y) \mid Y = y)\right] \right. \\ 
&\hspace{1.5cm} - \left. \E_{y \sim \distribution{Y}}\E_{x \sim \distribution{X_i} }\left[ \H(f(X, Y) \mid X_i = x, Y = y)\right]\right) \tag{Independence of $I$ and $(X, Y)$}\\
&= \frac{1}{n} \cdot \sum_{i \in [n]} \H(f(X, Y) \mid Y) - \H(f(X, Y) \mid X_i, Y)\tag{\autoref{def:cond-ent}} \\
 &= \frac{1}{n} \cdot \sum_{i \in [n]} \I(X_i ; f(X, Y)\mid Y) \tag{\autoref{def:mi}}\\
 &\leq  \frac{1}{n} \cdot \sum_{i \in [n]} \I(X_i ; f(X, Y)\mid Y, X_{<i}) + \I(X_i ;  X_{<i} \mid Y)  \tag{\autoref{lemma:info} } \\
 &=  \frac{1}{n} \cdot \sum_{i \in [n]} \I(X_i ; f(X, Y)\mid Y, X_{<i})  \tag{\Cref{item:fact:info3}, \autoref{fact:info}}\\
 &=  \frac{1}{n} \cdot  \I(X ; f(X, Y)\mid Y)  \tag{\Cref{item:fact:info4}, \autoref{fact:info}}.
\end{align*}
\end{proof}

\subsection{Measures of Distance Between Distributions.}

We use two main measures of distance (or divergence) between distributions, namely the Kullback-Leibler divergence (KL-divergence) and the total variation distance. 

\begin{definition}[KL-divergence]
\label{def:kl}
For two distributions $\mu$ and $\nu$ over the same set $\Omega$, the Kullback-Leibler divergence between $\mu$ and $\nu$, denoted by $\kl{\mu}{\nu}$, is defined as 
\[\kl{\mu}{\nu} = \sum_{x \in \Omega} \mu(x) \log\frac{\mu(x)}{\nu(x)}.\]
\end{definition}

\begin{definition}[Total Variation Distance]
\label{def:tvd}
For two distributions $\mu$ and $\nu$ over the same set $\Omega$, the total variation distance $\mu$ and $\nu$ is defined as 
\[\tvd{\mu}{\nu}:= \max_{\Omega' \subseteq \Omega}  \sum_{x \in \Omega'} \mu(x)- \nu(x).\] 
\end{definition}

These definitions satisfy the following properties:
\begin{fact}
\label{fact:kltvd}
The following hold:
\begin{enumerate}
\item \label{item:fact:kltvd1} For discrete random variables $X$, $Y$, and $Z$, we have 
\[\I(X ; Y \mid Z) = \E_{(y,z) \sim \distribution{(Y,Z)} } \left[ \kl{\distribution{X \mid Y=y, Z=z}}{\distribution{X \mid Z=z}}\right].\]
\item \label{item:fact:kltvd2}{\em (Pinsker's inequality)} For any distributions $\mu$ and $\nu$, we have 
\[
\tvd{\mu}{\nu} \leq \sqrt{\frac{1}{2} \cdot \kl{\mu}{\nu}}.
\] 
\end{enumerate}
\end{fact}

\section{Omitted Proofs}\label{app:proofs}
\begin{proof}[Proof of \autoref{thm:mainformal} assuming \autoref{thm:mainred}]
Proof by contradiction. Suppose that \autoref{thm:mainred} is true and \autoref{thm:mainformal} is not. Let $P(\cdot)$ be the polynomial promised by \autoref{thm:dobzinski} and let $d$ be the degree of $P$. Define $\beta = \frac{1}{500(d+1)}$. 
Let $\varepsilon_{\star} > 0$ be the constant promised by the negation of \autoref{thm:mainformal} for this value of $\beta$ (recall that we assume that \autoref{thm:mainformal} is false).
Let $m_1$ be large enough so that 
\begin{inparaenum}[(1)]
\item $P(m') \leq  m'^{d+1}$ for all $m' > m_1$,
\item $\exp(\beta \varepsilon_{\star}^2 \cdot m') \geq m'$ for all $m' > m_1$,
\item $m_1 > \frac{10^{10}}{\varepsilon_{\star}^2}$.
\end{inparaenum}

Using our assumption that \autoref{thm:mainformal} is false, we get that there is an  $m > m_1$, and a randomized, $m$-item, $\xos_m$-combinatorial auction $\Pi$ with two bidders and one seller that is truthful, is  $\left(\frac{3}{4} - \frac{1}{240} + \varepsilon_{\star} \right)$-approximate with probability $\frac{1}{2} + \exp( - \beta \varepsilon_{\star}^2  \cdot m)$, and satisfies $\cc(\Pi) < \exp(\beta \varepsilon_{\star}^2 \cdot  m)$. 

Plugging $\Pi$ into \autoref{thm:dobzinski}, we get a randomized, $m$-item, $\xos_m$-combinatorial auction $\Pi'$ with two bidders and one seller that is simultaneous and  $\left(\frac{3}{4} - \frac{1}{240} + \varepsilon_{\star} \right)$-approximate with probability $\frac{1}{2} + \exp( - \beta \varepsilon_{\star}^2  \cdot  m) > \frac{1}{2} + \exp\left( - \frac{ \varepsilon_{\star}^2 m }{500}\right)$ and satisfies (using $m > m_1$)
\[\cc(\Pi') < P(\max(\exp(\beta \varepsilon_{\star}^2 \cdot  m), m)) \leq  \exp\left( \frac{ \varepsilon_{\star}^2 m }{500}\right).\]
This contradicts \autoref{thm:mainred} and we are done. 
\end{proof}

\subsection{Omitted Proofs from~\Cref{sec:partition}}\label{app:partition}
\paragraph{Concentration inequalities.}

 We use the following version of Chernoff bound for negatively correlated random variables:

\begin{definition}[Negatively Correlated Random Variables]
\label{def:negcor}
For $n >0$, let $X_1, \cdots, X_n$ be random variables taking values in $\{0,1\}$. The random variables $X_1, \cdots, X_n$ are \emph{negatively correlated} if for all subsets $S \subseteq [n]$, we have $\Pr(\forall i \in S: X_i = 1) \leq \prod_{i \in S} \Pr(X_i = 1)$.
\end{definition}

\begin{lemma}[Generalized Chernoff Bound; cf.~\cite{PanconesiS97}]\label{lemma:negcorchernoff} 
For $n >0$, let $X_1, \cdots, X_n$ be negatively correlated random variables that take values in $\{0,1\}$. Then, for any $\varepsilon > 0$, we have (where $\mu =  \sum_{i \in [n]} \E[X_i] \leq n$):
\[
\Pr\left(\sum_{i \in [n]} X_i > \mu + \varepsilon n\right) \leq  \Pr\left(\sum_{i \in [n]} X_i >  (1 + \varepsilon ) \cdot \mu \right) \leq \exp(-\varepsilon^2  \mu /3).
\]
\end{lemma}

Much of the proofs in this section will follow by connecting $\pc(k,\vec{P},\vec{p})$ to a related product distribution, defined below.

\begin{definition}
For a partition parameter $(k, \vec{P}, \vec{p})$, define $\ally(k, \vec{P}, \vec{p})$ to be the distribution over subsets of $M$ such that we have $\Pr_{U \sim \ally(D)}(z \in U) = \frac{p_{\vec{P}[z]}}{\lvert{P_{\vec{P}[z]}}\rvert}$ independently for all $z \in M$.
\end{definition}

We will need the following technical  lemmas about partition parameters
\begin{lemma}
\label{lemma:pcdistally}
For any subset $S \subseteq M$ and any partition parameter $(k, \vec{P}, \vec{p})$, it holds that 
\[
\Pr_{U \sim \pc(k, \vec{P}, \vec{p})}(U \cap S = \emptyset) \leq \Pr_{U \sim \ally(k, \vec{P}, \vec{p})}(U \cap S = \emptyset). 
\]
\end{lemma}
\begin{proof}
We have 
\begin{align*}
\Pr_{U \sim \pc(k, \vec{P}, \vec{p})}(U \cap S = \emptyset) &= \frac{\lvert{\{U \subseteq \overline{S} \mid \lvert{\vec{P} \cap U}\rvert = \vec{p}\}}\rvert}{\lvert{\{U \subseteq M \mid \lvert{\vec{P} \cap U}\rvert = \vec{p}\}}\rvert}\\
&= \frac{\prod_{i \in [k] : \lvert{P_i}\rvert > 0}\binom{\lvert{\overline{S} \cap P_i}\rvert}{p_i} }{\prod_{i \in [k]: \lvert{P_i}\rvert > 0}\binom{\lvert{P_i}\rvert}{p_i}}\\
&= \prod_{i \in [k]: \lvert{P_i}\rvert > 0}\frac{\left( \lvert{P_i}\rvert - p_i\right) \left(\lvert{P_i}\rvert - p_i  - 1 \right) \cdots\left(\lvert{\overline{S} \cap P_i}\rvert - p_i + 1 \right)}{\lvert{P_i}\rvert \left(\lvert{P_i}\rvert -1 \right) \cdots\left(\lvert{\overline{S} \cap P_i}\rvert + 1 \right)}\\
&\leq \prod_{i \in [k]: \lvert{P_i}\rvert > 0}\left(1 - \frac{p_i}{\lvert{P_i}\rvert }\right)^{\lvert{S \cap P_i}\rvert} = \Pr_{U \sim \ally(k, \vec{P}, \vec{p})}(U \cap S = \emptyset).
\end{align*}
\end{proof}

\begin{corollary}
\label{cor:pcdistally}
For any partition parameter $(k, \vec{P}, \vec{p})$ and any distribution $D^*$ over subsets of $M$, it holds that 
\[
\Pr_{\substack{U \sim \pc(k, \vec{P}, \vec{p}) \\ U^* \sim D^*}}(U \cap U^* = \emptyset) \leq \Pr_{\substack{U \sim \ally(k, \vec{P}, \vec{p})\\ U^* \sim D^*}}(U \cap U^* = \emptyset). 
\]
\end{corollary}
\begin{proof}
We have:
\begin{align*}
\Pr_{\substack{U \sim \pc(k, \vec{P}, \vec{p}) \\ U^* \sim D^*}}(U \cap U^* = \emptyset) &=  \sum_{S \subseteq M}\Pr_{\substack{U \sim \pc(k, \vec{P}, \vec{p}) \\ U^* \sim D^*}}(U \cap S = \emptyset, U^* = S)\\
&=  \sum_{S \subseteq M}\Pr_{U \sim \pc(k, \vec{P}, \vec{p})}(U \cap S = \emptyset)\Pr_{U^* \sim D^*}(U^* = S)\\
&\leq  \sum_{S \subseteq M}\Pr_{U \sim \ally(k, \vec{P}, \vec{p})}(U \cap S = \emptyset)\Pr_{U^* \sim D^*}(U^* = S) \tag{\autoref{lemma:pcdistally}}\\
&=  \sum_{S \subseteq M}\Pr_{\substack{U \sim \ally(k, \vec{P}, \vec{p}) \\ U^* \sim D^*}}(U \cap S = \emptyset, U^* = S)  =\Pr_{\substack{U \sim \ally(k, \vec{P}, \vec{p}) \\ U^* \sim D^*}}( U \cap U^* = \emptyset). 
\end{align*}
\end{proof}

\begin{proof}[Proof of~\Cref{lemma:pcdist}]
Let $D$ denote the partition parameter $(k, \vec{P}, \vec{p}) $ and $D'$ denote the parameter $(k', \vec{P'}, \vec{p'})$. 
Let $U$ and $U'$ be sets sampled from distributions $\pc(D)$ and $\pc(D')  $ respectively. 
For $z \in M$, we define the indicator random variable $X_z$ to be such that $X_z = 1$ if and only if $z \notin U \cap U'$. We have that 
\begin{equation}
\label{eq:pcdist1}
\begin{split}
\E[X_z] = \Pr(X_z = 1) &=  \Pr_{\substack{U \sim \pc(D) \\ U' \sim \pc(D')}}(z \notin U \cap U') = 1 -  \Pr_{\substack{U \sim \pc(D) \\ U' \sim \pc(D')}}(z \in U \cap U') \\
&= 1 -  \Pr_{U \sim \pc(D)}(z \in U) \cdot \Pr_{U' \sim \pc(D')}(z \in U') =  1 - \frac{p_{\vec{P}[z]}}{\lvert{P_{\vec{P}[z]}}\rvert} \cdot \frac{p'_{\vec{P'}[z]}}{\lvert{P'_{\vec{P'}[z]}}\rvert},
\end{split}
\end{equation}
implying 
\(
\sum_{z \in [m]} \E[X_z] = m - \sum_{z \in M}\frac{p_{\vec{P}[z]}}{\lvert{P_{\vec{P}[z]}}\rvert} \cdot \frac{p'_{\vec{P'}[z]}}{\lvert{P'_{\vec{P'}[z]}}\rvert}. = m - \Delta
\).
 We now show that the random variables $X_1, \cdots, X_m$ are negatively correlated (\autoref{def:negcor}), whence it follows from \autoref{lemma:negcorchernoff} that 
 \begin{align*}
\Pr_{\substack{U \sim \pc(D) \\ U' \sim \pc(D')}} \left(\lvert{U \cap U'}\rvert < \Delta - \varepsilon m\right) &= \Pr\left(\sum_{z \in [m]} X_z >\sum_{z \in [m]} \E[X_z] + \varepsilon m  \right) \leq   \exp(-\varepsilon^2 (m -  \Delta)/3 ).
 \end{align*} 
 
 In order to show that the random variables $X_1, \cdots, X_m$ are negatively correlated, we pick an arbitrary subset $S$ of $M$ and show that $\Pr(\forall z \in S: X_z = 1) \leq \prod_{z \in S} \Pr(X_z = 1)$. We have:
\begin{align*}
\Pr(\forall z \in S: X_z = 1) &= \Pr_{\substack{U \sim \pc(D) \\ U' \sim \pc(D') }} \left(S  \cap U \cap U' = \emptyset \right) \\
&\leq \Pr_{\substack{U \sim \ally(D) \\ U' \sim \pc(D')}} \left(S  \cap U \cap U' = \emptyset \right) \tag{\autoref{cor:pcdistally}}\\
&\leq \Pr_{\substack{U \sim \ally(D) \\ U' \sim \ally(D')}} \left(S  \cap U \cap U' = \emptyset \right) \tag{\autoref{cor:pcdistally}}\\
&= \Pr_{\substack{U \sim \ally(D) \\ U' \sim \ally(D')}} \left(\forall z \in S: z \notin U \cap U' \right) \\
&= \prod_{z \in S}\Pr_{\substack{U \sim \ally(D) \\ U' \sim \ally(D')}} \left(z \notin U \cap U' \right) \\
&= \prod_{z \in S}\left(1 - \frac{p_{\vec{P}[z]}}{\lvert{P_{\vec{P}[z]}}\rvert} \cdot \frac{p'_{\vec{P'}[z]}}{\lvert{P'_{\vec{P'}[z]}}\rvert} \right)\\
&= \prod_{z \in S} \Pr(X_z = 1). \tag{\autoref{eq:pcdist1}}
\end{align*}

\end{proof}

\begin{proof}[Proof of~\Cref{lemma:part}] We only argue for the case $j = 1$ as the case $j = 2$ is symmetric. Let $\mathcal{C}$ be the set of all sequences $\vec{Z}'$ of $k_1$ subsets of $M$ satisfying $\lvert{\part_{\vec{S} \Vert \vec{Z}'}}\rvert = \vec{a}_1$. If $\vec{Z} \notin \mathcal{C}$, then the result holds as both the terms are $0$. We, thus assume that $\vec{Z} \in \mathcal{C}$. We immediately get $\Pr_{\vec{S}_1 \sim \mu_1}\left(\vec{S}_1 = \vec{Z}\right)  = \frac{1}{\lvert{\mathcal{C}}\rvert}$.

For $\vec{Z}' \in \mathcal{C}$, define the set $\mathcal{D}(\vec{Z}')$ to be the set of all sequences  $\vec{Z}''$ of $k_2$ subsets of $M$ such that $\lvert{\part_{\vec{S} \Vert \vec{Z}' \Vert \vec{Z}''}}\rvert = \vec{a}$. Owing to the fact that $\Pr_{\vec{S}_1 \sim \mu_1, \vec{S}_2 \sim \mu_2}\left(\lvert{\part_{\vec{S} \Vert \vec{S}_1 \Vert \vec{S}_2}}\rvert = \vec{a}\right) > 0$, we have  $\lvert{\part_{\vec{S} \Vert \vec{Z}''}}\rvert = \vec{a}_2$  for all $\vec{Z}'' \in \mathcal{D}(\vec{Z}')$. Furthermore, by symmetry, the value of $\lvert{\mathcal{D}(\vec{Z}')}\rvert$ is the same for all $\vec{Z}' \in \mathcal{C}$.

It follows that
\[ \Pr_{\substack{\vec{S}_1 \sim \mu_1\\ \vec{S}_2 \sim \mu_2}}\left(\vec{S}_1 = \vec{Z} \mid \lvert{\part_{\vec{S} \Vert \vec{S}_1 \Vert \vec{S}_2}}\rvert = \vec{a}\right) = \frac{\lvert{\mathcal{D}(\vec{Z})}\rvert}{\sum_{\vec{Z}' \in \mathcal{C}} \lvert{\mathcal{D}(\vec{Z}')}\rvert} = \frac{1}{\lvert{\mathcal{C}}\rvert},
\]
finishing the proof.
\end{proof}

\begin{proof}[Proof of~\Cref{cor:part}]

Observe that there exist unique $\vec{a}'_1 = \vec{a}'_1(\vec{S}, \vec{a}_1)$ and $ \vec{a}'_2 =  \vec{a}'_2(\vec{S}, \vec{a}_2)$, both in $ \mathbb{Z}^{2^{k+1}}$ such that, for any  $j \in \{1, 2\}$ and $A \subseteq M$, 
\[
\lvert{\part_{\vec{S}} \cap A }\rvert = \vec{a}_j \iff \lvert{\part_{\vec{S} \Vert A}}\rvert = \vec{a}'_j.
\]
Similarly, for any $\vec{a}$ such that $\Pr_{A_1 \sim \mu_1, A_2 \sim \mu_2}\left(\lvert{\part_{\vec{S}} \cap A_1 \cap A_2 }\rvert = \vec{a}\right) > 0$, there exists a unique $\vec{a}' = \vec{a}' (\vec{S}, \vec{a}_1, \vec{a}_2, \vec{a}) \in \mathbb{Z}^{2^{k+2}}$ such that, for all $A_1, A_2$ such that $\lvert{\part_{\vec{S}} \cap A_j }\rvert = \vec{a}_j$ for $j \in [2]$, we have, 
\[
\lvert{\part_{\vec{S}} \cap A_1 \cap A_2 }\rvert = \vec{a} \iff \lvert{\part_{\vec{S} \Vert A_1 \Vert A_2} }\rvert = \vec{a}'.
\] 
The proof then follows by applying \autoref{lemma:part} with  $k_1 = k_2 = 1$.
\end{proof}

\end{document}